\documentclass[USenglish]{article}

\usepackage[utf8]{inputenc}
\usepackage[big]{dg_arxiv}

\usepackage{subfigure}
\usepackage{ifpdf}
\usepackage{epstopdf}
\usepackage{color}
\usepackage{bbm}

\theoremstyle{plain}
\newtheorem{theorem}{Theorem}[section]
\newtheorem{cor}[theorem]{Corollary}
\newtheorem{lemma}[theorem]{Lemma}
\newtheorem{proposition}[theorem]{Proposition}

\theoremstyle{definition}

\numberwithin{equation}{section} 

\renewcommand{\paragraph}[1]{{\bf #1.}}
\definecolor{myred}{rgb}{0.8,0,0}  
\definecolor{mygreen}{rgb}{0,0.7,0}  
\newcommand{\neu}{\color{blue}}
\renewcommand{\neu}{}
\newcommand{\neuneu}{\color{mygreen}}
\renewcommand{\neuneu}{}
\newcommand{\mymarginpar}[1]{ \marginpar{{\tiny #1}}}
\renewcommand{\mymarginpar}[1]{} 
{\vskip\baselineskip\noindent\textbf{Proof of {#1}:}}%
{\hspace*{.1pt}\hspace*{\fill}\BOX\vskip\baselineskip}

\def \R{\mathbb{R}}               
\def \N{\mathbb{N}}               
\def \1{{\bf 1}}                
\def \0{{\bf 0}}

\def\qed{\hfill$\Box$}

\def\eps{\varepsilon}

\def\revlevel{\overline{\mu}}

\def\revspeed{\kappa}

\def\driftinitial{\overline{\filter}_0}
\def\filterinitial{\filter_0}
\def\covinitial{\overline{\variance}_0}
\def\condcovinitial{\variance_0}

\def\vol_R{\sigma_R}

\def\filter{m}
\def\variance{q}

\def\revspeed{\kappa}

\def\HF{F}
\def\HR{R}
\def\HD{J}
\def\HC{Z}
\def\nAktien{d}

\def\nWienerRendite{d_1}
\def\nWienerDrift{d_2}
\def\nWienerExperten{d_3}

\def\varianceexp{\Gamma}
\def\reliable{\Gamma}

\def\volR{\sigma_R}
\def\voldrift{\sigma_{\mu}}

\def\filter{m}
\def\variance{q}

\def\gammaq{\gamma}

\def\drift{\mu}

\def\komppoi{\widetilde{N}^{\lambda}}

\def\Mpro{M}
\def\Qpro{Q}

\def \R{\mathbb{R}}               
\def \N{\mathbb{N}}               
\def \1{{\bf 1}}                
\def \0{{\bf 0}}

\def\qed{\hfill$\Box$}

\def\eps{\varepsilon}

\def\vol_R{\sigma_R}

\newcommand\norm[1]{\left\lVert#1\right\rVert}
\def \alphaQD{\alpha^{\HD,\lambda}}
\def \alphaQC{\alpha^{\HC,\lambda}}
\def \gammah{\psi}

\def \alphaQCbar{\overline{\alpha}^{\lambda}}
\def \myzeta{J}
\def\uppbound{K}
\def\uppboundC{\uppbound^\HC}
\def\uppboundD{\uppbound^\HD}
\def\mynu{r}

\def\varexp{\Gamma}

\def \trace{\operatorname{tr}}
\def \Erw{\mathbb{E}}
\def \diag{\operatorname{diag}}

\begin{document}
	
  \articletype{~}

  \author[1]{Abdelali Gabih}
  \author[2]{Hakam Kondakji}
  \author[3]{Ralf Wunderlich}
  \runningauthor{Gabih, Kondakji and Wunderlich}
  \affil[1]{
    Laboratoire Informatique et Mathematiques
    et leurs Applications (LIMA), Faculty of Science, Chouaib Doukkali
    University, El Jadida 24000,  Morocco,    \texttt{a.gabih@uca.ma}
    }
  \affil[2]{Institute of Computer Science, Brandenburg University of Technology \newline Cottbus - Senftenberg, Postbox 101344, 03013 Cottbus, Germany, \texttt{hakam.kondakji@b-tu.de }}
  \affil[3]{Institute of Mathematics, Brandenburg University of Technology \newline Cottbus - Senftenberg, Postbox 101344, 03013 Cottbus, Germany, \texttt{ralf.wunderlich@b-tu.de }\newline
  The authors thank Dorothee Westphal and Jörn Sass  (TU Kaiserslautern) for valuable discussions that improved this paper.    }

  \title{ Asymptotic Filter Behavior for High-Frequency Expert Opinions in a Market
  with  Gaussian    Drift
}
  \runningtitle{ Asymptotic Filter Behavior for High-Frequency Expert Opinions}
  \subtitle{~}
  \abstract{
                This paper investigates a financial market where stock returns depend on a hidden
        Gaussian mean reverting drift process. Information on the drift is obtained from returns and
        expert opinions in the form of noisy signals about the current state of the drift arriving at the jump times of a homogeneous Poisson process.  Drift estimates are based on Kalman filter techniques and described by the conditional mean and  covariance matrix of the drift given the observations.
        We study the  filter asymptotics for increasing arrival intensity of expert opinions and prove that the conditional mean is a consistent  drift estimator, it  converges in the mean-square sense to the hidden drift. Thus, in the limit as the arrival intensity goes to infinity investors have full information about the drift.
        }

  \journalname{}
\journalyear{~}
  \journalvolume{}
  \journalissue{}
  \startpage{1}

  \maketitle


\section{Introduction}
In this paper we investigate  a hidden Gaussian model (HGM) for a
financial market where asset prices follow a diffusion process with
an unobservable Gaussian mean reverting drift  modelled by an
Ornstein-Uhlenbeck process. Such models are widely used in the study
of portfolio optimization problems under partial information on the
drift. There are two popular model classes for the drift, the {\neuneu above-mentioned}
HGM and {\neuneu the} hidden Markov model (HMM)  {\neuneu in which} the drift process is a
continuous-time Markov chain. For utility maximization problems
under HGM we refer to Lakner \cite{Lakner (1998)} and  Brendle
\cite{Brendle2006} while HMMs are used in Rieder and Bäuerle
\cite{Rieder_Baeuerle2005}, Sass and Haussmann \cite{Sass and
Haussmann (2004)}. Both models are studied in Putsch\"ogl and Sass
\cite{Putschoegl and Sass (2008)}. A generalization of these
approaches and further references can be found  in Björk et al.
\cite{Bjoerk et al (2010)}.

For solving  portfolio problems under partial information the drift
has to be estimated from  observable quantities such as
stock returns.
For the above two models, HGM and HMM, the conditional distribution
of the drift process given the return observations  can be described
completely by finite-dimensional filter processes. This allows for
efficient solutions to portfolio problems including the computation
of an optimal policy.  For HGM and HMM  finite-dimensional filters
are known as the Kalman and Wonham filters, respectively, see e.g.\
Elliott, Aggoun and Moore \cite{Elliott et al. (1994)}, Liptser and
Shiryaev \cite{Liptser-Shiryaev}.

It is well known that the drift of a diffusion process is
particularly hard to estimate. Even the estimation of a constant
drift would require empirical data over an extremely large time
horizon, see Rogers \cite[Chapter 4.2]{Rogers (2013)}. Therefore, in
practice filters computed from historical price observations lead to
drift estimates of quite poor precision since drifts  tend to
fluctuate randomly over time  and drift effects are overshadowed by
volatility. At the same time optimal investment strategies in
dynamic  portfolio optimization depend crucially on the drift of the
underlying asset price process. For these reasons, practitioners
also {\neuneu utilize} external sources of information such as news,
company reports, ratings or their own intuitive views on the future
asset performance for the construction of  optimal portfolio
strategies. These outside sources of information  are called
\textit{expert opinions}. The idea goes back to the celebrated
Black-Litterman model which is an extension  of the classical
one-period Markowitz model, see Black and  Litterman
\cite{Black_Litterman (1992)}. It uses  Bayesian updating to improve
drift estimates.

Contrary to the classical static one-period model, we  consider a
 continuous-time model  for asset prices  where
additional information {\neuneu in the form of} expert opinions arrives repeatedly over time. Davis and Lleo \cite{Davis and Lleo (2018)} termed
that approach  ``Black-Litterman in Continuous-Time'' (BLCT).  {\neuneu The} first
papers addressing  BLCT  are   Frey et al. (2012) \cite{Frey et al.
(2012)} and their follow-up paper \cite{Frey-Wunderlich-2014}. They
consider an HMM for the drift and expert opinions arriving at the
jump times of a Poisson process and study the maximization of
expected power utility of terminal wealth.  An HGM and expert
opinions arriving at fixed and known times have been   investigated
in Gabih et al. \cite{Gabih et al (2014)} for a market with only one
risky stock, and generalized in Sass et al. \cite{Sass et al (2017)}
for markets with multiple risky stocks. Here, the authors consider
maximization of logarithmic utility. Davis and Lleo \cite{Davis and
Lleo (2013), Davis and Lleo (2018)} consider BLCT for power utility
maximization under  an HGM and expert opinions arriving continuously
{\neuneu over} time. This allows for quite explicit solutions for the portfolio
optimization problem.  In \cite{Davis and Lleo (2018)}, the authors
also focus on the calibration of the model for  expert opinions
to real-world data.

In a recent paper Sass et al. \cite{Sass et al (2018)} consider an
HGM with expert opinions both  at fixed  as well as random
information dates  and investigate the asymptotic behavior of the
filter for increasing  arrival frequency of the expert opinions.
They assume that  a higher frequency of expert opinions is only
available at the cost of accuracy. In particular,  the variance of
expert opinions grows linearly with the  arrival frequency. This
assumption reflects that it is not possible for investors to gain
{\neuneu unlimited  information in a finite} time interval.  Furthermore,
it allows to find a certain asymptotic behavior that yields  reasonable filter approximations
 for investors observing {\neuneu  discrete-time expert opinions arriving with a fixed and sufficiently large intensity}.
The authors derive limit theorems which state that the information obtained from
observing high-frequency discrete-time expert opinions is
asymptotically the same as that from observing a certain diffusion
process {\neuneu that has} the same drift as the return process. The latter
process can be interpreted as a {continuous-time expert} which
permanently delivers noisy information about the drift. These
so-called diffusion approximations show how the  BLCT model of Davis
and Lleo \cite{Davis and Lleo (2013), Davis and Lleo (2018)}  who
work with continuous-time expert opinions can be obtained as a limit
of BLCT models with discrete-time experts.

The present paper can be considered as a companion paper to the
above mentioned work of Sass et al. \cite{Sass et al (2018)}.
However, contrary to \cite{Sass et al (2018)} we assume that the
expert's reliability expressed by its variance remains  {\neu bounded}
when the arrival intensity increases. {\neu For the sake of simplicity we restrict to the case of a constant variance.} This leads to a different
asymptotic regime corresponding to the Law of Large Numbers while
the results in \cite{Sass et al (2018)} are  in the sense of
Functional Central Limit Theorems.

When the arrival intensity increases,  the investor receives more and more  noisy signals about the
current state of the drift of the same precision. {\neuneu It is then} expected
that in the limit the drift estimate is perfectly accurate and
equals the actual drift, i.e., the investor has full information
about the  drift. While this statistical consistency of the
estimator seems to be intuitively clear a rigorous proof is an open
issue and will be addressed in this paper. Gabih et al. \cite{Gabih
et al (2014)} and  Sass et al. \cite{Sass et al (2017)} provide such
 proof only for the case of fixed and known information dates.
However,  their results and  methods {\neuneu cannot} be applied to the
present model with random information dates.  Note that also the
methods for the proof of the diffusion limits in \cite{Sass et al
(2018)} do not carry over to the present case of fixed expert's
reliability. To the best of our knowledge the techniques for proving
convergence constitute a new contribution to the literature.
Compared to  \cite{Gabih et al (2014)} and   \cite{Sass et al
(2017)} we do not only give a rigorous convergence proof but  we are
also able to determine the rate of convergence and give 
explicit bounds for  the estimation error.

In this paper we concentrate on the asymptotic properties of drift
estimates which are based on Kalman filter techniques and described
by the conditional mean and  covariance matrix of the drift given the
observations. We show that for increasing arrival intensity of
expert opinions, the  expectation of the  conditional covariance goes to zero.
This implies that  the conditional mean is a consistent  drift
estimator, it  converges  to the hidden
drift in the mean-square sense. We expect that these convergence results carry over to the
value functions of portfolio optimization problems but do not
include these studies in this paper. For the maximization of
expected logarithmic utility, the convergence of value functions
 has {\neuneu already} been proven in Sass et al.~\cite{Sass et al (2018)}. The
case of power utility will be addressed   in our  follow-up paper
\cite{Gabih et al (2018-2)}.

\medskip
The paper is organized as follows: In Section \ref{market_model} we
introduce the model for our financial market including  expert
opinions and define information regimes for investors with different
sources of information. For each of those information regimes, we
state  the dynamics of the corresponding
conditional mean and conditional covariance process in Section \ref{Filtering}. Section
\ref{asymptotic_filter} contains our main contributions  and studies the asymptotic filter behavior  for increasing arrival intensity of discrete-time expert opinions.
First, Lemma \ref{alpha_Properties}  gives an estimate for the drift
term in the semimartingale representation  of the conditional
covariance process. Based on this estimate, Theorem
\ref{Convergence_trace} shows  that {\neuneu as the arrival intensity increases}
the expectation of the  conditional covariance goes to zero. As a
consequence, Theorem \ref{Convergence_filter} states  the mean-square
convergence of the conditional mean  to the hidden drift. 
In Section \ref{asymptotic_filter_cont}, we study a related problem  for continuous-time expert opinions {\neuneu that arises in the case of} diffusion approximations of  discrete-time expert opinions.
Section \ref{numerics}  illustrates the convergence results by some
numerical experiments. In Appendix~\ref{appendix} we collect some
auxiliary results and technical  proofs needed for our main
theorems.

\medskip
\paragraph{Notation}  Throughout this paper, we use the notation $I_d$ for the identity matrix in $\R^{d\times d}$. For a symmetric and positive-semidefinite matrix $A\in\R^{d\times d}$ we call a symmetric and positive-semidefinite matrix $B\in\R^{d\times d}$ the \emph{square root} of $A$ if $B^2=A$. The square root is unique and will be denoted by $A^{\frac{1}{2}}$.

For a vector $X$ we denote by $\norm{X}$ the Euclidean norm.  For a square
matrix  $A$ we denote by 
 $\norm{A}$ a generic matrix norm, by $\norm{A}_F=\sqrt{\sum_{i,j}(A^{ij})^2}$  the
 Frobenius norm and by  $\trace(A)=\sum_i A^{ii}$
the trace of $A$.

\section{Financial Market}
\label{market_model}
\subsection{Price Dynamics}
 The  setting is based on  Gabih et
al.~\cite{Gabih et al (2014)} and Sass et al.~\cite{Sass et al
(2017), Sass et al (2018)}. For a  fixed date $T>0$ representing the
investment horizon, we work on a filtered probability space
$(\Omega,\mathcal{G},\mathbb{G},P)$, with filtration
$\mathbb{G}=(\mathcal {G}_t)_{t \in [0,T]}$ satisfying the usual
conditions. All processes are assumed to be $\mathbb{G}$-adapted.

 We consider a market model  for one risk-free bond with
 constant risk-free interest rate 
and $\nAktien$ risky securities whose return process
$R=(R^{1},\ldots,R^{\nAktien})$ is defined by
\begin{align}
dR_t=\mu_t\; dt+\volR\; dW^{R}_t, \label{ReturnPro}
\end{align}
for a given $\nWienerRendite$-dimensional $\mathbb{G}$-adapted
Brownian motion $W^{\HR}$ with $\nWienerRendite\ge d$. The  constant volatility matrix $\volR\in\mathbb
R^{\nAktien\times\nWienerRendite}$ is  assumed to be   such that $\Sigma_{R}:=\volR\volR^{\top}$ is positive definite.
In this setting the price process $S=(S^1,\ldots,S^{\nAktien})$ of
the risky securities reads as
\begin{align}
dS_t&=diag(S_t)\, dR_t. \label{stockmodel}
\end{align}
Note that we can write
\begin{align}
\label{Aktienpreis_2}
\log S_t^{i}&=\log s_0^{i}+ \int\limits_0^t \drift_s^{i}ds +\sum\limits_{j=1}^{\nWienerRendite}\Big( \sigma_R^{ij}W_t^{R,j}-\frac{1}{2} (\sigma_R^{ij})^2 t\Big) \nonumber\\
&=\log s_0^{i} +
R_t^{i}-\frac{1}{2}\sum\limits_{j=1}^{\nWienerRendite}
(\sigma_R^{ij})^2 t ,\quad i=1,\ldots,\nAktien.
\end{align}
So we have the equality $\mathbb{G}^R = \mathbb{G}^{\log S} =
\mathbb{G}^S$, where  for a generic process $X$ we denote by $\mathbb{G}^X$ the filtration generated by $X$. This is useful since it allows to work with $R$
instead of $S$ in the filtering part.

 The dynamics of the drift process $\mu=(\mu_t)_{t\in[0,T]}$ in \eqref{ReturnPro}
 are
given by the stochastic differential equation (SDE)
\begin{eqnarray}
\label{drift} d\mu_t=\revspeed(\revlevel-\mu_t) dt+\voldrift
dW^{\mu}_t,
\end{eqnarray}
where $\revspeed\in\mathbb R^{\nAktien\times\nAktien}$,
$\voldrift\in\mathbb R^{\nAktien\times\nWienerDrift}$ and
$\revlevel\in\mathbb R^{\nAktien} $ are constants such that the
matrices $\revspeed$ and $\Sigma_{\mu}:=\voldrift\voldrift^{\top}$
are positive definite, and $W^{\mu}$ is a
$\nWienerDrift$-dimensional Brownian motion independent of $W^{R}$ with  $\nWienerDrift\ge d$.
Here, $\revlevel$ is the mean-reversion level, $\revspeed$ the
mean-reversion speed and $\voldrift$ describes the volatility of
$\mu$. The initial value $\drift_0$ is assumed to be a normally
distributed random variable independent of $W^{\mu}$ and $W^{R}$
with mean $\driftinitial\in \R^{\nAktien}$ and covariance matrix
$\covinitial\in\mathbb R^{\nAktien\times\nAktien}$ assumed to be
symmetric and  positive semidefinite.  It is well known that  SDE
\eqref{drift} has the closed-form solution
\begin{eqnarray}
\label{mu_explicit} \mu_t = \revlevel +e^{-\revspeed t}\Big[(\mu_0
-\revlevel) +   \int_0^t e^{\revspeed s} \voldrift
dW^{\mu}_s\Big],\quad t\ge 0.
\end{eqnarray}
This is a  Gaussian process and  known  as  Ornstein-Uhlenbeck
process. It has mean value and covariance function
\begin{align*}
\Erw[\drift_t]&=\revlevel+e^{-\revspeed
t}(\driftinitial-\revlevel)\quad \text{and}\quad\\
\operatorname{Cov}(\drift_s,\drift_t)&=e^{-\revspeed s}
\Bigg(\covinitial+\int\limits_0^{\min\{s,t\}} e^{\revspeed u}
\Sigma_{\drift} e^{\revspeed^{\!\top} u} du \Bigg)
e^{-\revspeed^{\!\top} t},\quad s,t\ge 0.
\end{align*}

\subsection{Expert Opinions}
\label{Expert_Opinions}  We assume that investors observe the return
process $R$ but they neither observe the factor process $\mu$ nor
the Brownian motion $W^{R}$. They do however know the model
parameters such as $\vol_R,\revspeed, \revlevel, \voldrift $  and
the distribution $\mathcal{N}(\driftinitial,\covinitial)$ of  the
initial value $\drift_0$. Information about the drift $\mu$ can be
drawn from observing the returns $R$. A special feature of our model
is that investors may also have access to additional information
about the drift in form of \textit{expert opinions} such as news,
company reports, ratings or their own intuitive views on the future
asset performance. The expert opinions provide  noisy signals about
the current state of the drift arriving at discrete points in time
$T_k$. We model these expert opinions  by a marked point process
$(T_k,Z_k)_k$, so that at $T_k$ the investor observes the
realization of a random vector $Z_k$ whose distribution depends on
the current state $\mu_{T_k}$ of the drift process. The arrival
dates $T_k$ are modelled as jump times of a standard Poisson process
with intensity $\lambda>0$, independent of {\neu both the Brownian motions $W^\HR, W^\myzeta$ and the initial value of the drift  $\mu_0$}, so that the timing
of the information arrival
does not carry any useful information about the drift. 
For the sake of convenience we also write $T_0:=0$ although no
expert opinion arrives at time $t=0$.

The signals or ``the expert views'' at time $T_k$ are modelled by
$\R^\nAktien$-valued  Gaussian random vectors
$Z_k=(Z_k^1,\cdots,Z_k^{\nAktien})^{\top}$ with
\begin{align}
\label{Expertenmeinungen_fest}
Z_k=\drift_{T_k}+{\varianceexp}^{\frac{1}{2}}\varepsilon_k,\quad
k=1,2,\ldots,
\end{align}
 where the matrix  $\varianceexp\in\R^{\nAktien\times\nAktien}$ is
symmetric and positive definite.
Further,  $(\varepsilon_k)_{k\ge 1}$ is a sequence of independent
standard normally distributed random vectors, i.e.,
$\varepsilon_k\sim \mathcal{N}(0,I_d)$. It is  also independent of
the Brownian motions $W^R, W^\mu$,  the initial value $\mu_0$
{\neu and the arrival dates $(T_k)_{k\ge 1}$}. That means that, given $\mu_{T_k}$, the expert
opinion $Z_k$ is $\mathcal{N}(\mu_{T_k},\varianceexp)$-distributed.
So, $Z_k$ can be considered as an unbiased estimate of the unknown
state of the drift at time $T_k$. The matrix $\reliable$ is a
measure of the expert's reliability. In a model with $\nAktien=1$
risky asset $\reliable$ is just the variance of the expert's
estimate of the drift at time $T_k$: the larger $\varexp$ the less
reliable is the expert.

Note that one may also allow for relative expert views where experts
give an estimate for the difference in the drift of two stocks
instead of absolute views. This extension   is studied in
Sch\"ottle et al.~\cite{Schoettle et al. (2010)} where the authors
show how to switch between these two models for expert opinions by
means of a pick matrix.

Finally, we introduce expert opinions arriving continuously over time. This is motivated by the results of 
 Sass et al. \cite{Sass et al (2018)}. There  the authors consider the information drawn from   observing  certain sequences of expert opinions and  show that for a large number of expert opinions it is essentially the same as the information resulting  from observing another diffusion process. 
 The interpretation of that  diffusion process is an expert providing continuous-time estimates about the state of the drift.  Let this estimate be given by the diffusion process
\begin{equation}\label{continuous_expert}
 d\myzeta_t =\mu_t\,d t +\sigma_\myzeta\,d W^\myzeta_t,
\end{equation}
where $W^\myzeta$ is an $\nWienerExperten$-dimensional Brownian motion with $\nWienerExperten\geq d$ that is independent of all other Brownian motions in the model and of the information dates $T_k$. 
The constant matrix $\sigma_\myzeta\in\R^{d\times \nWienerExperten}$ is  assumed to be  such that $\Sigma_{\myzeta}:=\sigma_\myzeta\sigma_\myzeta^{\top}$ is positive definite.

\subsection{Investor Filtration}
\label{Investor_Filtration}  We consider various types of investors
with different levels of information. The information available to
an investor is described by the \textit{investor filtration}
$\mathbb{F}^H=(\mathcal{F}^H_t)_{t\in[0,T]}$. Here, $H$ denotes the
information regime for which we
 consider the cases $H=\HR,\HC,\HD,\HF$, where
\[\begin{array}{rcll}
\mathbb{F}^{\HR}&=& (\mathcal {F}_t^{\HR})_{t \in [0,T]} & \text{with }\mathcal {F}_t^{\HR}=\sigma(R_s,~ s\le t), \\[0.5ex]
\mathbb{F}^{\HC}&=& (\mathcal {F}_t^{\HC})_{t \in [0,T]} & \text{with }\mathcal {F}_t^{\HC}=\sigma(R_s, s\le t,\,~(T_k,Z_k),~ T_k\le t), \\[0.5ex]
\mathbb{F}^{\HD}&=& (\mathcal {F}_t^{\HD})_{t \in [0,T]} &  \text{with }\mathcal {F}_t^{\HD}=\sigma(R_s,\myzeta_s,~ s\le t), \\[0.5ex]
 \mathbb{F}^F&=&  (\mathcal {F}_t^{F})_{t \in [0,T]} &  \text{with }\mathcal {F}_t^{F}=\sigma(R_s, \mu_s,~ s\le t).
\end{array}
\]
We assume that the above $\sigma$-algebras $\mathcal{F}_t^H$
 are augmented by the null sets 
 of $P$. 
{\neu We call the investor with filtration $\mathbb{F}^H=(\mathcal{F}^H_t)_{t\in[0,T]}$ the $H$-investor.
The $\HR$-investor observes only the return process $R$, the $\HC$-investor combines
return observations with the discrete-time expert opinions
$Z_k$ while  the $\HD$-investor observes the return process together with the continuous-time expert $\myzeta$. Finally,  the $F$-investor has full information  and can observe the drift process $\mu$.} 
For stochastic
drift this case is not realistic, but we use it as a benchmark  and
in the next section  it will serve as a limiting case for
high-frequency expert opinions. We will denote an investor  with investor filtration $\mathbb{F}^H$ as $H$-investor. 

We assume that at $t=0$ the  partially informed investors
start with the same initial information given by the
$\sigma$-algebra $\mathcal{F}_0^I$, i.e., $\mathcal{F}_0^H=\mathcal
{F}_0^I\subset \mathcal F_0^{\HF}$, $H=\HR,\HC,\HD$. This initial
information $\mathcal{F}_0^I$ models prior knowledge about the drift
process at time $t=0$, e.g., from observing  returns  or expert
opinions in the past before the trading period $[0,T]$. We assume
that the conditional distribution of the initial drift value $\mu_0$
given $\mathcal{F}_0^H$ is the normal distribution
$\mathcal{N}(m_0,q_0)$ with mean $\filterinitial\in \R^{\nAktien}$
and covariance matrix $\condcovinitial\in\mathbb
R^{\nAktien\times\nAktien}$ assumed to be symmetric and  positive
semidefinite.
 In this setting typical examples are:
\begin{enumerate}
\renewcommand{\labelenumi}{\alph{enumi})}
\item  The investor has no  information about the initial value of the drift $\mu_0$. However, he knows  the model parameters,
in particular the distribution
$\mathcal{N}(\driftinitial,\covinitial)$ of $\mu_0$ with  given
parameters $\driftinitial$ and $\covinitial$. This corresponds to
$\mathcal{F}_0^I=
\{\varnothing,\Omega\}$ and
 $\filterinitial=\driftinitial$, $\condcovinitial=\covinitial$.

\item  The investor can fully observe the initial value of the drift $\mu_0$, which corresponds
to $\mathcal{F}_0^I=\mathcal{F}_0^F$ and
$\filterinitial=\drift_0(\omega)$ and $\condcovinitial=0$.

\item  Between the above limiting cases we consider an investor who has some prior but no complete information about
$\mu_0$ leading to $\{\varnothing,\Omega\}\subset \mathcal
{F}_0^I\subset\mathcal F_0^F $.
\end{enumerate}

\section{Partial Information and Filtering}
\label{Filtering}  The trading decisions of investors are based on
their knowledge about the drift process $\mu$. While the
$F$-investor observes the drift directly, the $H$-investor for
$H=\HR,\HC,\HD$ has to estimate it. This leads us to a filtering problem
with hidden signal process $\mu$ and observations given by the
returns $R$ and expert opinions $(T_k,Z_k)$  or $J$. The \textit{filter} for
the drift  $\mu_t$ is its projection on the
$\mathcal{F}_t^H$-measurable random variables described by the
conditional distribution of the drift given $\mathcal{F}_t^H$. The
mean-square optimal estimator for the drift at time $t$, given the
available information is  the \textit{conditional mean}
    $$\Mpro_t^{H}:=\Erw[\mu_t|\mathcal{F}_t^H].$$
The accuracy of that estimator  can be described by the
\textit{conditional covariance matrix}
\begin{align}
\label{cond_variance__def}
\Qpro_t^{H}:=\Erw[(\mu_t-\Mpro^{H}_t)(\mu_t-\Mpro^{H}_t)^{\top}|\mathcal{F}^{H}_t].
\end{align}
Since in our filtering problem  the signal   $\mu$, the observations
and the initial value of the filter  are jointly Gaussian also the
filter distribution is Gaussian and completely characterized by the
conditional mean $\Mpro_t^{H}$ and the conditional covariance
$\Qpro_t^{H}$.

In  Section \ref{asymptotic_filter} we will study the asymptotic behavior of the
filter for the $\HC$-investor observing  expert opinions  arriving
more and more frequently and   derive limit theorems for the filter
if the  arrival intensity $\lambda$ tends to infinity. 
Section \ref{asymptotic_filter_cont} is devoted to a related problem and considers the asymptotics  of the	filter processes for the $\HD$-investor with volatility $\sigma_\myzeta $ tending to zero.
These results
are based on the following dynamics of the filters for $H=\HR,\HC,\HD$ which
already can be found in  Sass et al.~\cite{Sass et al (2017), Sass
et al (2018)}.

\subsection{$R$- and $\HD$-Investor}
The $R$-investor  only observes returns and has no access to
additional expert opinions,  the information is given by
$\mathbb{F}^R$. Then, we are in the classical case of the  Kalman
filter, see e.g.~Liptser and Shiryaev \cite{Liptser-Shiryaev},
Theorem $10.3$, leading to the following  dynamics of $\Mpro^R$ and
$\Qpro^R$.

\begin{lemma}
\label{Kalmann_Filter_R_lemma} For the $R$-investor the filter  is
Gaussian and the conditional distribution of the drift $\mu_t$ given
$\mathcal F_t^{R}$ is the
normal  distribution $\mathcal N\left(\Mpro_t^{R},\Qpro_t^{R}\right)$.\\
The conditional mean $\Mpro^{R}$ follows the dynamics
\begin{align}
\label{Filter_R}
d\Mpro_t^{R}&=\revspeed(\revlevel-\Mpro_t^{R})\;dt+\Qpro_t^{R}\,
{\neu \Sigma_R^{-1}\; \left(dR_t-\Mpro_t^\HR dt\right)}.
\end{align}
The dynamics of the conditional covariance  $\Qpro^{R}$ is given by
the   Riccati differential equation
\begin{align}
\label{Riccati_R} d\Qpro_t^{R}&=(\Sigma_{\mu}-\revspeed
\Qpro_t^{R}-\Qpro_t^{R} \revspeed^{\top}-\Qpro_t^{R} \Sigma_{R}^{-1}
\Qpro_t^{R})\; dt.
\end{align}
The initial values are $\Mpro_0^{R}=\filterinitial $ and
$\Qpro_0^{R}=\condcovinitial$.
\end{lemma}

Note that the conditional covariance matrix $\Qpro_t^{R}$ satisfies
an ordinary differential equation and is hence deterministic,
whereas the conditional mean $\Mpro_t^{R}$ is a stochastic process
defined by an SDE {\neu driven by the return process $R$. 

 Next, we consider the $\HD$-investor who observes a $2d$-dimensional diffusion process with components $R$ and $\myzeta$. That observation process is driven by a $(\nWienerRendite+\nWienerExperten)$-dimensional Brownian motion with components $W^\HR$ and $W^\myzeta$.}
Again, we can apply classical Kalman filter theory as in Liptser and Shiryaev \cite{Liptser-Shiryaev} to deduce the dynamics of $\Mpro^\HD$ and $\Qpro^{\HD}$. We also refer to Lemma 2.2 in the companion paper  \cite{Sass et al (2018)}.
 
\begin{lemma}
	\label{Kalmann_Filter_D_lemma} For the $\HD$-investor the filter  is
	Gaussian and the conditional distribution of the drift $\mu_t$ given
	$\mathcal F_t^{\HD}$ is the
	normal  distribution $\mathcal N\left(\Mpro_t^{\HD},\Qpro_t^{\HD}\right)$.\\
	The conditional mean $\Mpro^{\HD}$ follows the dynamics
	\begin{align}
	\label{Filter_D}
	d\Mpro_t^{\HD}&=\revspeed(\revlevel-\Mpro_t^{\HD})\;dt+\Qpro_t^{\HD}\,
	{\neu (\Sigma_R^{-1}, \Sigma_\myzeta^{-1})
	\begin{pmatrix} 
	dR_t-\Mpro_t^\HD	dt \\
	d\myzeta_t-\Mpro_t^\HD	dt
	\end{pmatrix}}.
	\end{align}
	The dynamics of the conditional covariance  $\Qpro^{\HD}$ is given by
	the  Riccati differential equation
	\begin{align}
	\label{Riccati_D} d\Qpro_t^{\HD}&=(\Sigma_{\mu}-\revspeed
	\Qpro_t^{\HD}-\Qpro_t^{\HD} \revspeed^{\top}-\Qpro_t^{\HD} {\neu (\Sigma_{R}^{-1} + \Sigma_{\myzeta}^{-1})}
	\Qpro_t^{\HD})\; dt. 	
	\end{align}
	The initial values are $\Mpro_0^{\HD}=\filterinitial $ and 	$\Qpro_0^{\HD}=\condcovinitial$.
\end{lemma} 
 
Note that,  as in case of the $R$-investor, the conditional covariance  $\Qpro^{\HD}$ is deterministic.

\subsection{$\HC$-Investor}
Now we consider the filter for the $\HC$-investor who
 combines continuous-time observations of  stock
returns and expert opinions received at discrete points in time.
\begin{lemma}
\label{filter_C}For the $\HC$-investor the filter  is Gaussian and the
conditional distribution of the drift $\mu_t$ given $\mathcal
F_t^{\HC}$ is the normal  distribution $\mathcal
N\left(\Mpro_t^{\HC},\Qpro_t^{\HC}\right)$.
\ \\[-2ex]
\begin{enumerate}
\item[(i)]
Between two information dates $T_k$ and $T_{k+1}$, $k\in\N_0$,   the
conditional mean  $\Mpro_t^\HC$  satisfies SDE \eqref{Filter_R}, i.e.,
\begin{align*}
d\Mpro_t^{\HC}&=\revspeed(\revlevel-\Mpro_t^{\HC})\;dt+\Qpro_t^{\HC}\,
{\neu \Sigma_R^{-1}\;\left(dR_t-\Mpro_t^\HC dt\right)} \quad\text{for}~~ t\in
[T_k,T_{k+1}).
\end{align*}
The conditional covariance $\Qpro^{\HC}$ satisfies the ordinary Riccati
differential equation \eqref{Riccati_R}, i.e.,
\begin{align*}
d\Qpro_t^{\HC}&=(\Sigma_{\mu}-\revspeed
\Qpro_t^{\HC}-\Qpro_t^{\HC} \revspeed^{\top}-\Qpro_t^{\HC} \Sigma_{R}^{-1}
\Qpro_t^{\HC})\; dt.
\end{align*}
The initial values are $\Mpro_{T_k}^\HC $ and $\Qpro^\HC_{T_k}$,
respectively, with $\Mpro_{0}^{\HC}=\filterinitial$ and~
$\Qpro_{0}^{\HC}=\condcovinitial$.
\item[(ii)]
At the information dates $T_k$,  $k\in\N$, the conditional mean  and
covariance $\Mpro_{T_k}^\HC$ and $\Qpro_{T_k}^\HC$ are obtained from the
corresponding values at time $T_{k^-}$ (before the arrival of the
view) using the update formulas
\begin{align*}
\Mpro_{T_k}^{\HC}&=\rho_k\Mpro_{T_k-}^{\HC}+(I_d-\rho_k)Z_k,\\[1.5ex]
\Qpro_{T_k}^{\HC}&=\rho_k \Qpro_{T_k-}^{\HC},
\end{align*}
with the update factor
$\rho_k=\reliable(\Qpro_{T_k-}^{\HC}+\reliable)^{-1}$.
\end{enumerate}
\end{lemma}
\begin{proof}
For a detailed proof we refer to  Lemma 2.3 in \cite{Sass et al
(2017)} and  Lemma 2.3 in \cite{Sass et al (2018)}.
\end{proof}
Note that the dynamics of $\Mpro^\HC$ and $\Qpro^\HC$ between
information dates are the same as for the $R$-investor, see Lemma
\ref{Kalmann_Filter_R_lemma}. The values at an information date
$T_k$ are obtained from a Bayesian update.

Recall that  for the $R$-investor the conditional mean $\Mpro^R$ is
a diffusion process and the conditional covariance $\Qpro^R$ is
deterministic. Contrary to that  the conditional mean $\Mpro^\HC$ of
the  $\HC$-investor is a jump-diffusion process and the conditional
covariance $\Qpro^\HC$ is no longer deterministic since the updates
lead to jumps at the random arrival dates  $T_k$ of the expert
opinions. Hence, $\Qpro^\HC$ is a piecewise deterministic stochastic
process.

\subsection{Properties of the Filter}
 The next lemma states in  mathematical terms  the intuitive
property that additional information from the expert opinions
improves drift estimates. Since the accuracy of the filter is
measured by the conditional covariance it is expected that this
quantity for the $\HC$-investor who combines observations of returns
and  expert opinions  is ``smaller'' than for the $R$-investor who
observes returns only.
Mathematically, this can be expressed by the   partial ordering of
symmetric matrices. For symmetric matrices $A,B\in\R^{d\times d}$ we
write $A \preceq B$ if $B-A$ is positive semidefinite. Note that $A
\preceq B$   implies that $\norm{A}\le \norm{B}$.

\begin{proposition}
\label{properties_filter} It holds $\Qpro^\HC_t \preceq \Qpro^R_t$ and $\Qpro^\HD_t \preceq \Qpro^R_t$. In particular, there exists a constant $C_{\Qpro}>0 $ such that 
$\norm{\Qpro^\HC_t} \le \norm{\Qpro^R_t}\le C_{\Qpro} $ and $\norm{\Qpro^\HD_t} \le \norm{\Qpro^R_t}\le C_{\Qpro} $ for all $t\in[0,T]$.
\end{proposition}
For the proof we refer to \cite{Sass et al (2018)}, Lemma 2.4.

\section{Filter Asymptotics for  High-Frequency Expert Opinions }
\label{asymptotic_filter}
 In the following we consider the $\HC$-investor and its filter for
increasing arrival intensity $\lambda$ and study the  asymptotic
behavior of the conditional mean   and conditional covariance  for
$\lambda\to \infty$. Then the average number of expert opinions per
unit of time goes to infinity, i.e., the $\HC$-investor has more and
more noisy estimates of the current state of the hidden drift at his
disposal. This will   lead to an increasing accuracy of the drift
estimator. As a consequence of the Law of Large Numbers we expect
that in the limit for $\lambda \to \infty$ the  drift estimator
coincides with the drift. In fact we show in Theorem
\ref{Convergence_filter} that the  drift estimator given by the
conditional mean $\Mpro^\HC$ converges to the hidden drift $\mu$ in
the mean-square sense with  rate $1/\sqrt{\lambda}$. Thus,  $\Mpro^\HC$
is a consistent estimator for $\mu$ and in the limit the
$\HC$-investor has full information about the drift.

Note that there is another asymptotic regime if additional expert
opinions only come at the cost of accuracy described by the variance
$\Gamma$. Assuming that this variance grows linearly in the arrival
intensity  Sass et al. \cite{Sass et al (2018)} show that the
information  the $\HC$-investor obtains from observing
 discrete-time expert opinions is asymptotically the same as that from observing a certain
diffusion process. The latter  can be interpreted as a
continuous-time expert. The limit theorems obtained in \cite{Sass et
al (2018)} allow to derive so-called diffusion approximations of the
filter for high-frequency discrete-time expert opinions. They
constitute  a Functional Central Limit Theorem while the limit
theorems obtained below for the case of fixed variance $\Gamma$  can
be considered as  a Functional Law of Large Numbers.

\smallskip
In our notation we now want to emphasize the dependence of the
filter processes and the investor filtration on the intensity
$\lambda$ by adding the superscript $\lambda$. Thus,  we write
$\Mpro_t^{\HC,\lambda}, \Qpro_t^{\HC,\lambda}$ and $\mathbb
F^{\HC,\lambda}$.

\subsection{Conditional Variance}
\label{Cond_Var}
We now show that the expectation of the conditional covariance
process $\Qpro_t^{\HC,\lambda}$  for $\lambda\to\infty$ goes to zero.
{\neu For this purpose, it will be  useful to express the dynamics of $Q^{\HC,\lambda}$ given in Lemma \ref{filter_C} in a unified way that comprises both the behavior between information dates and the jumps at times $T_k$.  We therefore work with a Poisson random measure as in Cont and Tankov~\cite[Sec.~2.6]{cont_tankov_2004}.
	Let $E=[0,T]\times\R^d$ and let $U_k$, $k=1,2,\dots$, be a sequence of independent multivariate standard Gaussian random variables on $\R^d$. For any $I\in\mathcal{B}([0,T])$ and $B\in\mathcal{B}(\R^d)$ let
	\[ N(I\times B)=\sum_{k\colon T_k\in I} \mathbbm{1}_{\{U_k\in B\}} \]
	denote the number of jump times in $I$ where $U_k$ takes a value in $B$. Then  $N$ defines a Poisson random measure with a corresponding compensated measure $\widetilde{N}^\lambda(ds,du)=N(ds,du)-\lambda\,ds\,\varphi(u)\,du$, where $\varphi$ is the multivariate standard normal density on $\R^d$, see Cont and Tankov~\cite[Sec.~2.6.3]{cont_tankov_2004}.
}

{\neu 
The next lemma rewrites the dynamics of   $\Qpro^{\HC,\lambda}$
given in Lemma \ref{filter_C} and provides a semimartingale
representation  which is driven by the martingale
$\komppoi$. For a detailed proof and further explanations we refer to Westphal~\cite[Prop.~8.14]{westphal_2019} and Kondakji \cite[Sec.~3.1]{Kondakji (2019)}. 
\begin{lemma}
	\label{Darstellung_Filter_Int_C} The dynamics of
	the conditional covariance matrix $\Qpro^{\HC,\lambda}$ are given by
	\begin{align}
	\label{Riccati_CN_3} d\Qpro_t^{\HC,\lambda}&=\alphaQC(\Qpro_t^{\HC,\lambda})\;
	dt+\int\limits_{\mathbb R^{\nAktien}} \gammaq^{}(\Qpro_{s-}^{\HC,\lambda})\;
	\komppoi(ds,du), \quad \Qpro_0^{\HC,\lambda}=\condcovinitial,
	\end{align}
where 
	\begin{align}
	\label{alpha_Q_def}
	\alphaQC(\variance)&= \Sigma_{\mu}-\revspeed \variance-\variance\revspeed^{\top}-\variance \Sigma_{R}^{-1} \variance 
	-\lambda\variance\left(\variance+\varianceexp\right)^{-1}\variance,\\
	\label{gamma_Q_def}  \gammaq^{}(\variance)&  =
	-\variance\left(\variance+\varianceexp\right)^{-1}\variance.
	\end{align}
\end{lemma}

We rewrite the above  $\mathbb F^{\HC,\lambda}$-semimartingale decomposition   of $\Qpro^{\HC,\lambda}$ in integral form and obtain}
\begin{align}
\Qpro^{\HC,\lambda}_t &=\mathcal A_t^{\lambda}+\mathcal K_t^{\lambda}
\quad\text{for } t\in[0,T], \label{q_Ito}
\end{align}
with
\begin{align*}
\mathcal A_t^{\lambda}&
:=\condcovinitial+\int_0^t\alphaQC(\Qpro_s^{\HC,\lambda})\;
ds\quad\text{and}\quad \mathcal
K_t^{\lambda}:=\int_0^t\int\limits_{\mathbb R^{\nAktien}}
\gamma(\Qpro_{s-}^{\HC,\lambda}) \;\komppoi(ds,du).
\end{align*}
Since by Proposition
\ref{properties_filter} the conditional covariance
$\Qpro^{\HC,\lambda}_s$ is bounded on $[0,t]$ also   $\gammaq$ is
bounded and the jump process $\mathcal K^{\lambda}$ is an $\mathbb
F^\HC$-martingale  and hence $\Erw [\mathcal K_t^{\lambda}]=0$ and
\begin{align}
\label{EQ_EA} \Erw [\Qpro_t^{\HC,\lambda}]=\Erw [\mathcal
A_t^{\lambda}].
\end{align}

For the study of the asymptotic behavior of the conditional covariance
$\Qpro^{\HC,\lambda}$ we investigate  the   drift of the process
$\mathcal A^{\lambda}$ which is given by the non-linear
matrix-valued function $\alphaQC$. The following lemma
gives an estimate of the trace of    $\alphaQC(q)$ in
terms of a linear function of the trace of $q$. That estimate   will
play a crucial role for deriving the convergence result  in Theorem
\ref{Convergence_trace}.
\begin{lemma}(\textbf{Properties of $\alphaQC$})\\
\label{alpha_Properties}%
For the function $\alphaQC$ given
in \eqref{alpha_Q_def} there exist constants
$a_{\alpha},b_{\alpha}>0$ independent of $\lambda$ and there exists
$\lambda_0>0$ such that  for all symmetric and positive
semidefinite $q\in \R^{d\times d}$
\begin{align}
\trace\big(\alphaQC(q)\big)\leq
a_{\alpha}-\sqrt{\lambda}\, b_{\alpha} \trace(q),\quad
\text{for}\quad \lambda\geq \lambda_0. \label{alpha_G_Abs}
\end{align}
The above estimate holds for every
$a_\alpha>\trace(\Sigma_\mu),$
\begin{align}
\label{constants_ablambda}
 b_\alpha < \overline{b}_{\alpha} =\overline{b}_{\alpha}(a_\alpha)&:=2\sqrt{\frac{a_{\alpha}-\trace(\Sigma_\mu)}{\trace(\Gamma)}},~~\\
 \nonumber
\lambda_0 = {\lambda}_0(a_{\alpha},\beta_{\alpha})&:=
\bigg(\frac{\nAktien(a_{\alpha}-\trace(\Sigma_\mu))}{
     2\sqrt{\trace(\Gamma)(a_{\alpha}-\trace(\Sigma_\mu))}-b_{\alpha}\trace(\Gamma)}\bigg)^2.
\end{align}
\end{lemma}

The quite technical proof is given in Appendix
\ref{proof_lemma_alphaQ}.
 The following main theorem gives an upper bound for the
expectation of the  trace of
  $\Qpro^{\HC,\lambda}$ from which the convergence to zero can be deduced.

\begin{theorem}
\label{Convergence_trace}  For every $\delta\in(0,T]$ there exists
 $\lambda_Q>0$ such that
\begin{align}
\Erw\big[\trace\big(\Qpro_t^{\HC,\lambda} \big)\big]\leq
\frac{\uppboundC}{\sqrt{\lambda}}\quad \text{for }~ \lambda\geq
\lambda_Q,\quad t\in[\delta,T]  \quad\text{and}
\label{bound_trace}\\
\uppboundC=\uppboundC(\delta)=
\big(\trace(\Gamma)[\trace(\Sigma_\mu)+\trace(\condcovinitial)(e\,\delta)^{-1}]\big)^{1/2}
\label{bound_c0},
\end{align} 
where $e=\exp(1)$ denotes Euler's number.\\
In particular, it holds~~ $~\Erw\big[ \trace\big(\Qpro_t^{\HC,\lambda}
\big)\big]\to 0\quad\text{as }~ \lambda\to \infty~ \text{for all }~
t\in(0,T].$
\end{theorem}

\begin{proof}
 Let us define the function $g(t):=\Erw\big[ \trace (\Qpro_t^{\HC,\lambda})\big] $
for $t\in[0,T]$. Then using
  \eqref{q_Ito}, \eqref{EQ_EA} and  the linearity of the expectation and the trace operator yields
\begin{align}
\nonumber g(t)=\trace
(\Erw[ \Qpro_t^{\HC,\lambda}])=\trace(\condcovinitial)+\int_0^t\Erw\big[ \trace(\alphaQC(\Qpro_s^{\HC,\lambda}))\big]\;
ds.
\end{align}
 Since according to Proposition  \ref{properties_filter} the conditional covariance  $\Qpro^{\HC,\lambda}$ is bounded  and piecewise continuous the function
$g$ is piecewise differentiable and for its derivative it holds
$g^{\prime}(t)=\Erw\big[ \trace(\alphaQC(\Qpro_t^{\HC,\lambda}))\big]$.
Further we have  $g(0)=\trace(\condcovinitial)$. Lemma
\ref{alpha_Properties} implies that there are constants $a_\alpha,
b_\alpha, \lambda_0>0$ such that
\begin{align}
\nonumber g^{\prime}(t)\leq \Erw\big[ a_{\alpha}-\sqrt{\lambda}\,
b_{\alpha}\trace(\Qpro_t^{\HC,\lambda})\big]
&=a_{\alpha}-\sqrt{\lambda}\,
b_{\alpha} \Erw\big[ \trace(\Qpro_t^{\HC,\lambda})\big]\\
&=a_{\alpha}-\sqrt{\lambda}\, b_{\alpha} g(t) \qquad\text{for
}\lambda\ge \lambda_0. \label{nach_lemma}
\end{align}
 We now apply Gronwall's Lemma in differential form to obtain for
$t\in[\delta,T]$  and $\lambda\ge \lambda_0$
\begin{align}
\label{g_esti} g(t) &\leq g(0)e^{-\sqrt{\lambda}\,
b_{\alpha}t}+\frac{a_{\alpha}}{\sqrt{\lambda}\,
b_{\alpha}}(1-e^{-\sqrt{\lambda}\, b_{\alpha}t})  \leq {
\frac{1}{\sqrt{\lambda}}\,\big (
    h(\delta,\lambda, b_\alpha) +\frac{a_{\alpha}}{b_{\alpha}} \big) },
\end{align}
where $h(\delta,\lambda, b_\alpha):=\trace(\condcovinitial)
\sqrt{\lambda}\;e^{-\sqrt{\lambda}\,b_{\alpha} \delta}$. Next we
show how  for given  $\delta\in(0,T]$ we can choose the constants
$a_\alpha, b_\alpha, \lambda_Q>0$ such that  $h(\delta,\lambda,
b_\alpha) +{a_{\alpha}}/{b_{\alpha}} \le \uppboundC(\delta)$ for
$\lambda\ge \lambda_Q$ with the constant $\uppboundC(\delta)$ given in
\eqref{bound_c0}.  Consider for $\lambda\ge 0$ the function
$\lambda\mapsto f(\lambda)=h(\delta,\lambda, b_\alpha)$ for fixed
$\delta\in(0,T]$ and $b_\alpha\in (0,\overline{b}_\alpha)$, where
$\overline{b}_\alpha$ is given in \eqref{constants_ablambda}. The
function $f$ is non-negative, it holds $f(0)=0$ and $f(\lambda) \to
0$ for $\lambda\to\infty$. There is a unique maximum at
$\lambda^*=(b_{\alpha}\delta)^{-2}$ with $f(\lambda^*)=(e
\,b_{\alpha}\delta)^{-1}\trace(\condcovinitial)$. Hence for the last
term on the  r.h.s.~of \eqref{g_esti} we obtain
\begin{align}
\label{h_esti}
 h(\delta,\lambda, b_\alpha)  +\frac{a_{\alpha}}{b_{\alpha}} \le
\frac{1}{b_{\alpha}} (\trace(\condcovinitial)(e\,\delta)^{-1}
+a_\alpha) \quad \text{for } \lambda\ge \lambda_0. 
\end{align}
The latter
expression is decreasing in $b_\alpha$ and the minimum on
$(0,\overline{b}_\alpha]$ is attained for $b_\alpha
=\overline{b}_\alpha$. According to \eqref{constants_ablambda} this
selection leads to $\lambda_0=\infty$ which is not feasible and we
have to restrict to values $b_\alpha <\overline{b}_\alpha$. However,
we can achieve  the above mentioned minimal value by choosing
$b_\alpha =\overline{b}_\alpha-\eta$ with a sufficiently small
$\eta>0$ and $\lambda_Q\ge \min(\lambda_0,\lambda^*)$ such that
$$ h(\delta,\lambda, b_\alpha)  +\frac{a_{\alpha}}{b_{\alpha}} =
h(\delta,\lambda,\overline{b}_\alpha-\eta)
+\frac{a_{\alpha}}{\overline{b}_\alpha-\eta} \le
\frac{1}{\overline{b}_\alpha}
(\trace(\condcovinitial)(e\,\delta)^{-1} +a_\alpha) \quad \text{for
} \lambda\ge \lambda_Q.$$ To see this estimate we note that
$f(\lambda)$ is decreasing on $(\lambda^*,\infty)$ and tends to zero
for $\lambda \to\infty$. Hence
$h(\delta,\lambda,\overline{b}_\alpha-\eta) $ can be made
arbitrarily small by selecting $\lambda$ large enough.

Finally, we study the dependence of  the above estimate on
$a_\alpha$ and  take into account the definition of
$\overline{b}_\alpha$ given in  \eqref{constants_ablambda}, i.e.~we
consider the  function
$$a_\alpha \mapsto \frac{1}{\overline{b}_\alpha} (\trace(\condcovinitial)(e\,\delta)^{-1} +a_\alpha)
=\frac{\sqrt{\trace(\Gamma)}}{2\sqrt{a_{\alpha}-\trace(\Sigma_\mu)}}
(\trace(\condcovinitial)(e\,\delta)^{-1} +a_\alpha) 
$$
for $a_\alpha > \trace(\Sigma_\mu)$.
There is a unique minimizer at
$a_\alpha^*=2\trace(\Sigma_\mu)+\trace(\condcovinitial)(e\,\delta)^{-1}$
and the minimal value is given by $\uppboundC$ defined in \eqref{bound_c0}.
 This proves the first claim.

Since that inequality holds for all $\delta\in(0,T]$, the convergence
$\Erw\big[ \trace\big(\Qpro_t^{\HC,\lambda} \big)\big]\to 0 ~\text{for }
\lambda\to \infty$ holds for all $t\in(0,T]$.
\end{proof}

 From the above asymptotic properties for  the expectation of the
trace of $\Qpro^{\HC,\lambda}$  we can easily deduce analogous results
for the expectation of the   norm $\|\Qpro^{\HC,\lambda}\|$ of the
conditional covariance.
\begin{cor}
   \label{Convergence_Frobenius}

  For every   $\delta\in(0,T]$ and any matrix norm $\|\cdot\|$ there exist  constants $C, \lambda_Q>0$
    such that 
        \begin{align}
        \Erw\big[ \, \big\|{\Qpro_t^{\HC,\lambda} }\big\|^p\big]\leq
        \frac{C}{\sqrt{\lambda}}\quad \text{for }~ \lambda\geq
        \lambda_Q,\quad t\in[\delta,T]  \text{ and } p\ge 1. \label{bound_covariance}
        \end{align} 
        For the  Frobenius norm $\|\cdot\|_F$ the constant $C$ can be chosen as  $C=\uppboundC C_F^{p-1}$ where 
        $\uppboundC$ {\neuneu is} given in \eqref{bound_c0} and $C_F$ denotes the upper bound from Proposition \ref{properties_filter} for the Frobenius norm  $\|\Qpro^{\HC,\lambda}\|_F$. 
        \\[0.5ex]
        In particular, it holds~~
        $~ { \Erw\big[ }\,\big\|{\Qpro_t^{\HC,\lambda} }\big\|^p\big] \to
        0\quad\text{as }~ \lambda\rightarrow \infty \quad \text{for all }~
        t\in(0,T].$
\end{cor}
\begin{proof}
For the Frobenius norm  of a symmetric and positive semidefinite matrix $A$ it holds $\|A\|_F \le \trace(A)$ (see Lemma \ref{properties_Gram}, Inequality \eqref{Frob_trace}). Further,  Proposition \ref{properties_filter} implies  $\|A\|_F \le C_F$. Hence
 $$\|\Qpro^{\HC,\lambda}\|_F^p\leq C_F^{p-1} \|\Qpro^{\HC,\lambda}\|_F  \le
C_F^{p-1} \trace(\Qpro^{\HC,\lambda})$$%
and Theorem \ref{Convergence_trace} with
inequality \eqref{bound_trace} proves the claim.  The equivalence of
matrix norms implies the assertion  for other norms.
\end{proof}
\subsection{Conditional Mean}
 We are now in a position to state and prove a similar convergence
 result for the asymptotic behavior of the filter $\Mpro$.   The proof is based on the following identity which relates the mean-square error of the filter estimate to the conditional covariance.
\begin{lemma}
\label{EM_EQ_lemma} It holds
\begin{align}
\label{EM_EQ} \Erw\big[ \,\big\|\Mpro_t^{\HC,\lambda}-\drift_t\big\|^2\,
\big]  = \trace\big(\Erw \big[\Qpro_t^{\HC,\lambda}\big]  \big).
\end{align}
\end{lemma}

\begin{proof}
For the mean-square criterion from \eqref{EM_EQ} it holds
\begin{align}
\nonumber
\Erw \big[\big\|\Mpro_t^{\HC,\lambda}-\drift_t\big\|^2\big]& =\Erw 
\big[(\Mpro_t^{\HC,\lambda}-\drift_t)^{\top}(\Mpro_t^{\HC,\lambda}-\drift_t)\big]\\
\label{conv0}
&=\trace\big(\Erw \big[
(\Mpro_t^{\HC,\lambda}-\drift_t)(\Mpro_t^{\HC,\lambda}-\drift_t)^{\top}\big]
\big).
\end{align}

For the expectation in the last term  the tower law of conditional
expectation  and the definition of the conditional covariance in
\eqref{cond_variance__def} yields
\begin{align*}
\Erw\big[(\Mpro_t^{\HC,\lambda}-\drift_t)(\Mpro_t^{\HC,\lambda}-\drift_t)^{\top}\big]
&= \Erw\Big[\Erw
\big[(\Mpro_t^{\HC,\lambda}-\drift_t)(\Mpro_t^{\HC,\lambda}-\drift_t)^{\top}\big|
\mathcal{F}^{\HC,\lambda}_t\big]\Big]\\
& =\Erw\big[\Qpro_t^{\HC,\lambda}\big].
\end{align*}
Substituting into   \eqref{conv0} yields the assertion \eqref{EM_EQ}.
\end{proof}
\begin{theorem}
\label{Convergence_filter}
Let $\uppboundC, \lambda_Q$   be the constants given in Theorem \ref{Convergence_trace}. 
Then for every   $\delta\in(0,T]$
\begin{align}
\Erw\big[ \,\big\|\Mpro_t^{\HC,\lambda}-\drift_t\big\|^2\, \big]  ~\leq~
\frac{\uppboundC}{\sqrt{\lambda}}\quad \text{for }~ \lambda\geq
\lambda_Q,\quad t\in[\delta,T]. \label{bound_fltred_drift}
\end{align} 
In particular, it holds~~
$\Erw\big[ \,\big\|\Mpro_t^{\HC,\lambda}-\drift_t\big\|^2\, \big]
\rightarrow 0\quad\text{as }~ \lambda\rightarrow \infty \quad
\text{for all }~ t\in(0,T].$
\end{theorem}
\begin{proof}
 Using identity     \eqref{EM_EQ} from Lemma  \ref{EM_EQ_lemma}
and applying inequality \eqref{bound_trace} of Theorem
\ref{Convergence_trace} we
 obtain
\begin{align}
\Erw \big[\big\|\Mpro_t^{\HC,\lambda}-\drift_t\big\|^2\big]  
&=\trace\big(\Erw \big[\Qpro_t^{\HC,\lambda}\big]  \big)
 \leq
\frac{\uppboundC}{\sqrt{\lambda}}\quad \text{for }~ \lambda\geq \lambda_0,~
t\in[\delta,T].\nonumber
\end{align}
 Since the above inequality holds for all $\delta\in (0,T]$ we finally obtain the desired convergence of the filter $\Mpro_t^{\HC,\lambda}$ for $t\in(0,T])$
as $\lambda\to\infty$, i.e.
$$
\Erw\big[ \big\|\Mpro_t^{\HC,\lambda}-\drift_t\big\|^2\ \big]\rightarrow
0 \quad \text{as }\lambda\rightarrow\infty.
$$
\end{proof}

\section{Filter Asymptotics for  Continuous-Time  Expert Opinions }
\label{asymptotic_filter_cont}
In the preceding section we already mentioned  that there is another asymptotic regime if  the variance of the expert opinions
$\Gamma$ is not independent of the arrival intensity $\lambda$  but grows linearly in $\lambda$. We now want to establish some relations to the case of constant $\Gamma$ studied above.

Suppose that   $\Gamma=\Gamma^\lambda = \lambda \sigma_\myzeta\sigma_\myzeta^{\top}$ where  $\sigma_\myzeta$ is the volatility matrix of the  continuous-time expert opinion process 
$d\myzeta_t =\mu_t\,d t +\sigma_\myzeta\,d W^\myzeta_t$
defined in \eqref{continuous_expert}. There we introduced the $\HD$-investor who combines observations of stock returns with those of $\myzeta$ (instead of discrete-time expert opinions).
For that setting   Sass et al. \cite{Sass et al (2018)} show that the
information  the $\HC$-investor obtains from observing
discrete-time expert opinions is asymptotically the same as the information of the $\HD$-investor extracting from observing the 
diffusion process $\myzeta$  if the model of the expert's views $Z_k$  given in \eqref{Expertenmeinungen_fest} uses standard normally distributed random variables  $\eps_k$ defined by the increments of $W^\myzeta$ in the form $\eps_k=\sqrt{\lambda} (W^\myzeta_{k/\lambda}- W^\myzeta_{(k-1)/\lambda}), k\in\N$.
 In particular they prove the mean-square convergence of filter processes $\Mpro^\HC, \Qpro^\HC$ to the corresponding filter processes $\Mpro^\HD, \Qpro^\HD$ of the $\HD$-investor and also provide the corresponding error estimates.  These limit theorems justify so-called diffusion approximations of the filter for high-frequency discrete-time expert opinions to \textit{fixed  and  sufficiently large} variance $\overline \Gamma$ of the expert stating that the filter for the $\HC$-investor can be approximated by the filter of a $\HD$-investor with volatility matrix $\sigma_\myzeta =\sigma_\myzeta^{\lambda}=\frac{1}{\sqrt{\lambda}} \overline \Gamma^{1/2}$. 

Motivated by the results of the preceding section where we studied the filter asymptotics of the $\HC$-investor with  fixed expert's variance  $\Gamma=\overline \Gamma$ for $\lambda \to\infty$ we now want to study the asymptotics of the associated diffusion approximations. We therefore introduce a family of diffusion processes $(\myzeta^\lambda)_{\lambda>0}$ defined by
\begin{equation}\label{continuous_expert_lambda}
d\myzeta^\lambda_t =\mu_t\,d t +\frac{1}{\sqrt{\lambda}}\,\overline\sigma_\myzeta\,d W^\myzeta_t
\end{equation}
with a constant matrix $\overline \sigma_\myzeta\in\R^{d\times \nWienerExperten}$ chosen   such that $\overline \Sigma_{\myzeta}:=\overline\sigma_\myzeta \overline\sigma_\myzeta^{\top}$ is positive definite.  Then it holds $\Sigma_{\myzeta}=\Sigma_{\myzeta}^\lambda=\frac{1}{\lambda}\overline \Sigma_{\myzeta}$.
Since  $=\frac{1}{\sqrt{\lambda}} \overline\sigma_\myzeta\to 0$ for $\lambda\to\infty$ the limit case is not covered by the limit theorems in \cite{Sass et al (2018)} and the diffusion approximation degenerates. Nevertheless, from a statistical point of view there is a clear interpretation. In the limit the $\HD$-investor can perfectly reconstruct the hidden drift $\mu$ from observing the limiting process $\myzeta^\infty$ defined by the (deterministic) ODE $d\myzeta^\infty_t =\mu_t\,d t $ and has thus full information on the drift.

Below we provide  a precise mathematical meaning to that convergence to full information  and prove the corresponding limit theorems for the filter processes which are analogues to their  counterparts for high-frequency discrete-time experts in Section \ref{asymptotic_filter}. We also provide the associated  error bounds for the drift estimates of the $\HD$-investor.
It turns out that we can benefit a lot from the techniques developed in the proofs of Section \ref{asymptotic_filter}.

Starting point is the conditional covariance $\Qpro^\HD=\Qpro^{\HD,\lambda}$. Note that, contrary to the stochastic conditional covariance  $ \Qpro^{\HC,\lambda}$ of the $\HC$-investor, $\Qpro^{\HD,\lambda}$ is deterministic. According to  Lemma \ref{Riccati_D} it satisfies the Riccati  differential equation \eqref{Riccati_D} which we rewrite as 
	\begin{align}
	\label{Riccati_D_cont}
	d\Qpro_t^{\HD,\lambda} & =\alphaQD(\Qpro_t^{\HD,\lambda})\, dt,\quad  \Qpro_0^{\HD,\lambda}=\condcovinitial, \qquad \text{where}\\ 
	 \label{alpha_Q_def_cont}
	 \alphaQD(\variance)& = \Sigma_{\mu}-\revspeed \variance-\variance\revspeed^{\top}-\variance {\neu  (\Sigma_R^{-1} + \lambda \overline \Sigma_\myzeta^{-1})} \variance.	 
	\end{align}

\begin{lemma}(\textbf{Properties of $\alphaQD$})\\
	\label{alpha_Properties_cont}%
For the function $\alphaQD$ given
	in \eqref{alpha_Q_def_cont} there exist constants
	$a_{\alpha},b_{\alpha}>0$ independent of $\lambda$  such that  for all symmetric and positive
	semidefinite $q\in \R^{d\times d}$
	\begin{align}
	\trace\big(\alphaQD(q)\big)\leq 
	a_{\alpha}-\sqrt{\lambda}\, b_{\alpha} \trace(q),\quad
	\text{for}\quad \lambda>0. \label{alpha_G_Abs_cont}
	\end{align}
	The above estimate holds for 
	\begin{align}
	\label{constants_ablambda_cont}
	a_\alpha=\trace(\Sigma_\mu)+(d \trace(\overline\Sigma_\myzeta)\mynu)^{-1} \quad \text{and} \quad 
	b_\alpha= 2(d \trace(\overline\Sigma_\myzeta)\sqrt{\mynu})^{-1} 
	\end{align}	
	and every $\mynu>0$.
\end{lemma}
The  proof is given in Appendix \ref{proof_lemma_alphaQ_cont}.
Note that contrary to the corresponding estimate for $\alphaQC$ given in Lemma \ref{alpha_Properties} the above estimate for $\alphaQD$ is valid not only for sufficiently large $\lambda\ge \lambda_0>0$ but for all $\lambda >0$. 

The following  theorem provides an analogous result to Theorem \ref{Convergence_trace} and  gives an upper bound for the expectation of the  trace of
$\Qpro^{\HD,\lambda}$ from which the convergence to zero can be deduced.

\begin{theorem}
	\label{Convergence_trace_cont}  For every $\delta\in(0,T]$ and $\lambda>0$ it holds
	\begin{align}
	\trace\big(\Qpro_t^{\HD,\lambda} \big) \leq
	\frac{\uppboundD}{\sqrt{\lambda}}\quad \text{for }~  t\in[\delta,T]  \quad\text{where}
	\label{bound_trace_cont}\\
	\uppboundD=\uppboundD(\delta)=
	\big(d  \trace(\overline\Sigma_\myzeta) [\trace(\Sigma_\mu) + \trace(\condcovinitial) (e\,\delta)^{-1}] \big)^{1/2},
	\label{bound_c0_cont}
	\end{align} 
	where $e=\exp(1)$ denotes Euler's number.\\
	In particular, it holds~~ $~\trace\big(\Qpro_t^{\HD,\lambda}	\big)\to 0\quad\text{as }~ \lambda\to \infty~ \text{for all }~
	t\in(0,T].$
\end{theorem}
\begin{proof}
	Let us define the function $g(t):=\trace (\Qpro_t^{\HD,\lambda}) $
	for $t\in[0,T]$. Then using
	\eqref{Riccati_D_cont} and  the linearity of  {\neuneu the} trace operator yields
	$g(t) =\trace(\condcovinitial)+\int_0^t\trace(\alphaQD(\Qpro_s^{\HD,\lambda}))
	ds=\trace(\condcovinitial)+\int_0^t g(s)ds.$
	Analogous to the proof of \eqref{h_esti} in Theorem \ref{Convergence_trace}, i.e., applying  Lemma \ref{alpha_Properties_cont} and Gronwall's Lemma we obtain  for $t\in[\delta,T]$  and $\lambda>0$	 
	\begin{align}
	\label{g_esti_cont} 
	g(t) &\leq  \frac{1}{\sqrt{\lambda}}\,\frac{1}{b_{\alpha}} (\trace(\condcovinitial)(e\,\delta)^{-1}
		+a_\alpha). 
	\end{align}
Recall \eqref{constants_ablambda_cont} stating that   the constants $a_\alpha, b_\alpha$ can be chosen as
\begin{align}a_\alpha=a_\alpha(\mynu)=\trace(\Sigma_\mu)+(d \trace(\overline\Sigma_\myzeta)\mynu)^{-1} ~ \text{and} ~
b_\alpha= b_\alpha(\mynu)=2(d \trace(\overline\Sigma_\myzeta)\sqrt{\mynu})^{-1} 
\end{align}	
for any $\mynu>0$. We now choose $\mynu$ such that the r.h.s.~of \eqref{g_esti_cont} attains its minimum. The unique minimizer is found as $\mynu^* =(d \trace(\overline\Sigma_\myzeta) [\trace(\condcovinitial)(e\,\delta)^{-1}+ \trace(\Sigma_\mu)])^{-1}$ and the minimal value is given by $\uppboundD/\sqrt{\lambda}$ with $\uppboundD$
  defined in \eqref{bound_c0_cont}.	This proves the first claim. 	
	Since that inequality holds for all $\delta\in(0,T]$ the convergence
	$\trace(\Qpro_t^{\HC,\lambda} )\to 0 ~\text{for }~
	\lambda\to \infty$ holds for all $t\in(0,T]$.
\end{proof}

As in Section \ref{Cond_Var}  the above asymptotic properties for   the
 trace of $\Qpro^{\HD,\lambda}$  imply  analogous results  for its norm. The proof is analogous to the proof of Corollary \ref{Convergence_Frobenius}.
\begin{cor}
	\label{Convergence_Frobenius_cont}
	
	For every   $\delta\in(0,T]$ and any matrix norm $\|\cdot\|$ there  exists {\neuneu a constant} $C>0$
	such that 
	\begin{align}
	 \big\|{\Qpro_t^{\HD,\lambda} }\big\|^p\leq
	\frac{C}{\sqrt{\lambda}}\quad \text{for }~ \lambda>0,\quad t\in[\delta,T]  \text{ and } p\ge 1. \label{bound_covariance_cont}
	\end{align} 
	For the  Frobenius norm $\|\cdot\|_F$ the constant $C$ can be chosen as  $C=\uppboundD C_F^{p-1}$ where 
		$\uppboundD$ {\neuneu is} given in \eqref{bound_c0_cont} and $C_F$ denotes the upper bound from Proposition \ref{properties_filter} for the Frobenius norm  $\|\Qpro^{\HD,\lambda}\|_F$. 
	\\[0.5ex]
	In particular, it holds~~
	$~  \,\big\|{\Qpro_t^{\HD,\lambda} }\big\|^ p \to
	0\quad\text{as }~ \lambda\rightarrow \infty \quad \text{for all }~
	t\in(0,T].$
\end{cor}
Based on the limit theorem for the conditional variance we now can state the corresponding result for the convergence of the conditional mean $\Mpro^{\HD}=\Mpro^{\HD,\lambda}$. The proof is analogous to the proof of Theorem \ref{Convergence_filter}.
\begin{theorem}
	\label{Convergence_filter_cont}
	Let $\uppboundD$   be the constant given in Theorem \ref{Convergence_trace_cont}. 
	Then for every   $\delta\in(0,T]$
	\begin{align}
	\Erw\big[ \,\big\|\Mpro_t^{\HD,\lambda}-\drift_t\big\|^2\, \big]  ~\leq~
	\frac{\uppboundD}{\sqrt{\lambda}}\quad \text{for }~ \lambda>0,\quad t\in[\delta,T]. \label{bound_fltred_drift_cont}
	\end{align} 
	In particular, it holds~~
	$\Erw\big[ \,\big\|\Mpro_t^{\HD,\lambda}-\drift_t\big\|^2\, \big]
	\rightarrow 0\quad\text{as }~ \lambda\rightarrow \infty \quad
	\text{for all }~ t\in(0,T].$
\end{theorem}
Note that contrary to the corresponding estimate for $\Mpro^{\HC,\lambda}$ given in Theorem \ref{Convergence_filter} the above estimate for $\Mpro^{\HD,\lambda}$ is valid not only for sufficiently large $\lambda\ge \lambda_Q>0$ but for all $\lambda >0$.

\section{Numerical Example}
\label{numerics} In this section we illustrate the theoretical
findings of the previous sections by results of some  numerical
experiments. These experiments  are based on a  stock market model
where  the unobservable drift $\mu$  follows an Ornstein-Uhlenbeck
process as given in \eqref{drift} and \eqref{mu_explicit} whereas
the volatility is known and constant. For simplicity, we assume that
there is only one risky asset in the market, i.e. $\nAktien=1$. For
our numerical experiments  we use the model parameters given in
Table~\ref{parameter}.

The distribution of the initial value $\mu_0$ of the drift process
is assumed to be the stationary distribution of the
Ornstein-Uhlenbeck process, i.e., the limit of the marginal
distribution of $\mu_t$ for $t\to \infty$ which is known to be
Gaussian with mean $\driftinitial=\revlevel$ and  variance
$\covinitial=\frac{\voldrift^2}{2\revspeed}$.

\begin{table}[ht]

\begin{tabular}{|ll|r||ll|r|}
\hline
\rule{0mm}{3ex}%
 Mean  reversion level & $\revlevel$ & $0.1$ & Time horizon & $T$ & $1$ year    \\\hline
\rule{0mm}{3ex}%
 Mean reversion speed & $\revspeed$ & $3$ & Stock volatility & $\volR$ & $0.25$   \\\hline
\rule{0mm}{3ex}%
 Volatility & $\voldrift$ & $1$ & Expert's variance &$\varianceexp= \overline\Sigma_\myzeta$ &$0.05$\\\hline
\rule{0mm}{3ex}%
Initial value $\mu_0$: mean   & $\driftinitial=\revlevel$ & $0.1$ &
Filter: initial values & $m_0=\driftinitial$ & $0.1$ \\\hline
\rule[-1.5ex]{0mm}{5ex}%
\phantom{Initial value ~} variance  & $\covinitial=\frac{\voldrift^2}{2\revspeed}$ &
$0.1\overline{6}$ & & $\variance_0 = \covinitial $ &
$0.1\overline{6} $ \\\hline
\end{tabular}
\\[1ex]
 \centering \caption{\label{parameter}
 \small Model parameters  for numerical experiments
 }

\end{table}
The arrival dates of the expert opinions are modelled as jump times
of a Poisson process with intensity $\lambda$. Then the waiting
times between two information dates are exponentially distributed
with parameter $\lambda$  and the  investor receives until time $T$
on average  $\lambda T$ expert opinions. Recall that the expert's
views are modelled by
\begin{align*}
Z_k=\drift_{T_k}+{\varianceexp}^{\frac{1}{2}}\varepsilon_k^{},\quad
k\in\N,
\end{align*}
where $(\varepsilon_k)_{k\ge 1}$ is a sequence of independent
standard normally distributed random variables.

\begin{figure}[h]
	\hspace*{-3mm}
	\includegraphics[width=1.0\textwidth]{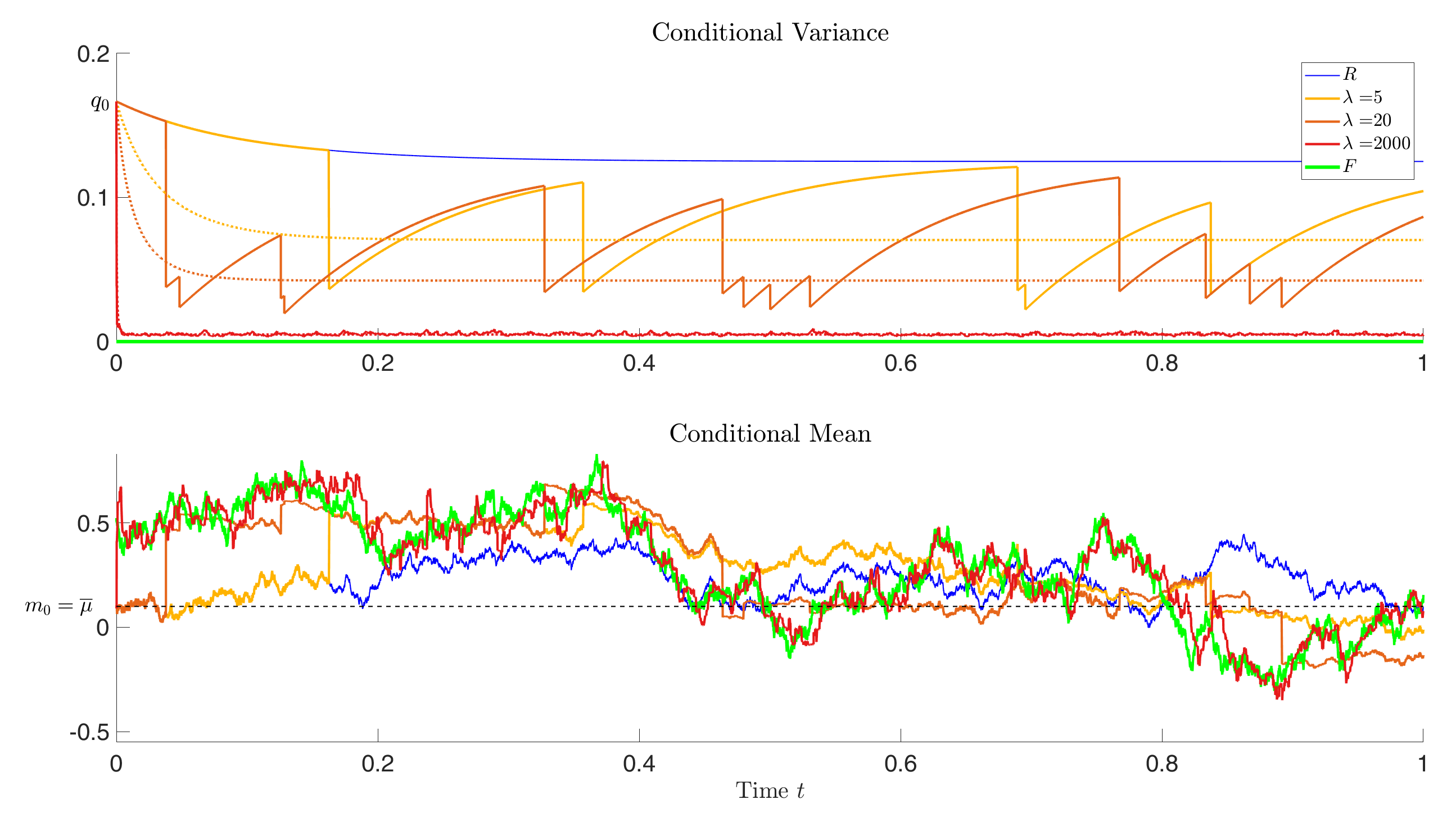}%
	
	\centering \caption{\label{high_frequency_opinions_Foto}
		\small
		Simulation of the filter processes $\Qpro^{H}$ and $\Mpro^{H}$. The upper subplot shows realizations of the conditional
		variances $\Qpro^{R}$,  $\Qpro^{\HC,\lambda}$ (solid) and $\Qpro^{\HD,\lambda}$ (dotted) for various
		intensities $\lambda$.  The volatility of the continuous-time expert opinions is chosen as $\sigma_\myzeta^{\lambda}=\sqrt{\Gamma/\lambda}$.
		The lower subplot shows  realizations of the
		corresponding conditional means $\Mpro^{R}$ and $\Mpro^{\HC,\lambda}$
		together with the path of the drift process $\mu$ (green).
	}
\end{figure}

At initial time $t=0$ all partially informed investors have the same
information about the hidden drift. For the experiment  we assume
that they only know the model parameters described by
$\mathcal{F}^H_0=\{\varnothing,\Omega\}$. Then the  initial values
for the filter processes  $\Mpro^H$ and $\Qpro^H$ are the parameters
of the Gaussian distribution of $\mu_0$,
i.e.~$m_0=\driftinitial=\revlevel$ and $\variance_0 =
\covinitial=\frac{\voldrift^2}{2\revspeed}$, respectively.

In Figure \ref{high_frequency_opinions_Foto} we plot the filters
given by conditional mean $\Mpro^H$ and conditional variance
$\Qpro^H$ of the $R$-investor  (blue), the   $\HC$-investor together with the associated $\HD$-investor against time. For the $\HC$-investor we consider the  arrival intensities
$\lambda=5, {\neu 20}, 2000$ (yellow, orange, red). The volatility of the associated continuous-time expert opinions is chosen as $\sigma_\myzeta^{\lambda}=\sqrt{\Gamma/\lambda}$.
In the upper plot one
can see the conditional variances $\Qpro^{R}$,   $\Qpro^{\HC,\lambda}$,  $\Qpro^{\HD,\lambda}$
and we also highlight (in green) the zero level corresponding to the
limit process for $\lambda\to\infty$. The lower plot shows a
realization of the unobservable   drift process $\drift$ (in green)
together with its estimates given by the conditional means
$\Mpro^{R}$ (blue) and  $\Mpro^{\HC,\lambda}$ (yellow, orange, red). We omit {\neuneu plotting} the paths of  $\Mpro^{\HD,\lambda}$.

Since the filter processes for the $R$- and $\HC$-investor start with
the same initial value their paths are identical until the arrival
of the first expert opinion leading to a filter update. This can be
nicely seen for $\lambda=5$ and also for $\lambda=20$ while for
$\lambda=2000$ the first update is almost immediately after the
initial time $t=0$. At the information dates the updates decrease
the conditional variance and lead to a jump of the conditional mean.
The updates of the conditional mean typically  decrease the distance
of $\Mpro^{\HC,\lambda}$ to the hidden drift $\mu$,  of course this
depends on the actual value of the expert's view. Note that the
drift estimate $\Mpro^R$ of the $R$-investor is quite poor and
fluctuates just around the mean-reversion level $\revlevel$.
However, the expert opinions  visibly improve the drift estimate.

After an update  the conditional variance $\Qpro^{\HC,\lambda}$
increases and if the waiting time to the next information date is
sufficiently large then it almost approaches the level of
$\Qpro^{R}$. Again, this can nicely be observed for $\lambda=5$.
During such long  periods without new expert opinions the
conditional mean of the $\HC$-investor  $\Mpro^{\HC,\lambda}$ tends to
move towards the path of $\Mpro^{R}$.

Looking at the paths of the conditional variance it can be seen that
$\Qpro^{R}_t$ dominates    $\Qpro^{\HC,\lambda}_t$ and  $\Qpro^{\HD,\lambda}_t$ for all $t\in(0,T]$
which confirms the corresponding property stated in Proposition
\ref{properties_filter} and illustrates the fact that additional
information by expert opinions  leads to improved  drift estimates.
Note that for increasing $t$ the conditional variances $\Qpro^R_t$  and  $\Qpro^{\HD,\lambda}_t$
quickly approach a constant which is the limit for $t\to\infty$.
That convergence   $\Qpro^R$ has been proven in Proposition 4.6 of Gabih et
al.~\cite{Gabih et al (2014)} for markets with a single stock and
generalized in Theorem 4.1 of Sass et al. \cite{Sass et al (2017)}
for markets with multiple stocks. The proof for $\Qpro^{\HD,\lambda}$ is analogous.

Comparing the paths of the filter processes of the
$\HC$-  and $\HD$-investor for increasing arrival intensity $\lambda$ it can be
observed that  the conditional variances $\Qpro^{\HC,\lambda}$ and  $\Qpro^{\HD,\lambda}$
approach zero for any $t\in(0,T]$. This fact illustrates our
findings in Theorems \ref{Convergence_trace}  and \ref{Convergence_trace_cont}.  Further, with
increasing $\lambda$ the  path of the conditional mean
$\Mpro^{\HC,\lambda}$ approaches the path of the hidden drift $\drift$
which confirms the mean-square convergence  stated in Theorem
\ref{Convergence_filter}.

\begin{figure}[h]
	\hspace*{-3mm}
	\includegraphics[width=1.0\textwidth,height=0.35\textwidth]{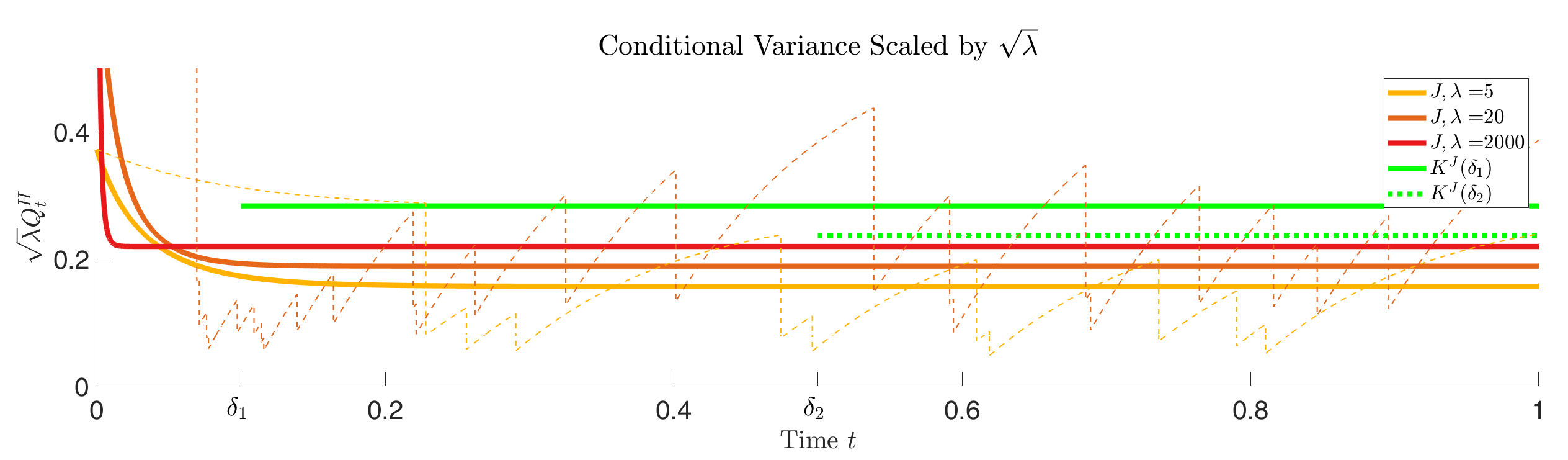}%
	
	\centering \caption{\label{variance_bounds}
		\small 		Conditional variances scaled by $\sqrt{\lambda}$ of the $\HD$-investor, i.e.~$\sqrt{\lambda}\,\Qpro^{\HD,\lambda}$ for $\lambda=5,20,2000$ (solid), and of the $\HC$-investor $\sqrt{\lambda}\,\Qpro^{\HC,\lambda}$ (dotted) for $\lambda=5,20$.  The volatility of the continuous-time expert opinions is chosen as $\sigma_\myzeta^{\lambda}=\sqrt{\Gamma/\lambda}$.
		The green lines represent the upper bounds $\uppboundD=\uppboundD(\delta)$ given in  Theorem \ref{Convergence_trace_cont} for $\delta=\delta_1=0.1$ and $\delta=\delta_2=0.5$.}
\end{figure}

Finally, we want to examine the goodness of the upper bounds $\uppboundC/\sqrt{\lambda}$ and $\uppboundD/\sqrt{\lambda}$  for the conditional variances of the $\HC$- and $\HD$-investor  given in 	Theorems \ref{Convergence_trace}  and \ref{Convergence_trace_cont}, respectively. Note that in the present example with $d=1$ stock the two constants $\uppboundC, \uppboundD$  coincide, it holds $\uppboundC= \uppboundD=
\big(\Gamma [\sigma_\mu^2 + \condcovinitial (e\,\delta)^{-1}] \big)^{1/2}.$ 
In order to facilitate the visual comparison of the conditional variances and their upper bounds we focus on the information regime $H=\HD$ and  rewrite the estimate  \eqref{bound_trace_cont} as  $\sqrt{\lambda}\Qpro_t^{\HD,\lambda}  \leq
\uppboundD =\uppboundD(\delta)~ \text{for }~  t\in[\delta,T]$. 
Fig.~\ref{variance_bounds} shows for $\lambda=5, {\neu 20}, 2000$ (yellow, orange, red solid lines) the conditional variances $\Qpro_t^{\HD,\lambda}$ scaled by $\sqrt{\lambda}$ together with the upper {\neuneu bounds} $\uppboundD(\delta)$ (green) for two values of $\delta$. It can be seen that the upper bounds are quite close to the actual values on $[\delta,T]$, in particular for larger $\delta$.

We also plot realizations of $\sqrt{\lambda}\Qpro_t^{\HC,\lambda}$ for the associated $\HC$-investor (dashed lines). Note that estimate \eqref{bound_trace} does hold for the expected variance $\Erw[ \Qpro_t^{\HC,\lambda}]$ but not for the realizations.

\begin{appendix}
    \section{Proofs }
    \label{appendix}
    \subsection{Auxiliary Results}
     The proof of Lemma \ref{alpha_Properties} which is given in Appendix \ref{proof_lemma_alphaQ}   is based on various
     properties of symmetric and positive semidefinite matrices which we
     collect in the next lemma.

         \begin{lemma}(\textbf{Properties of symmetric and positive semidefinite matrices})\\
        \label{properties_Gram}%
        Let $A,B\in\R^{d\times d}, d\in\N,$ symmetric and positive
        semidefinite matrices. Then it holds
        \begin{enumerate}
            \item  $A+B$ is symmetric positive semidefinite.
            \item The eigenvalues $\varrho_i=\varrho_i(A)$ of $A$ are
            nonnegative, and there exists an orthogonal matrix $V$ such that
            \begin{align}
            A=VDV^{\top}\quad\text{with}\quad
            D=\diag(\varrho_1,\cdots,\varrho_d),\label{diagona}
            \end{align}
            i.e., $A$ is diagonalizable.
            \item 
            \label{Ainv_symm_posdefinit}
            If $A$ is positive definite then it is nonsingular and the
            inverse  $A^{-1}$ is symmetric and positive definite.
            \item  
             \begin{align}
             \varrho_{\min}(A)\trace\big(B\big) \leq\trace(AB)\leq\varrho_{\max}(A)\trace(B)
             \label{trace-symmetric}
             \end{align}
              where $\varrho_{\min}(A)$ and $\varrho_{\max}(A)$ denote the smallest and largest eigenvalue of $A$, respectively.
            \item
            \begin{align}
             \frac{\trace(B)}{\trace\big(A^{-1}\big) } \trace(AB)\leq\trace(A)\trace(B)
            \label{trace-product}
            \end{align}
            where for the first inequality $A$ is assumed to be   positive definite.
            \item
            \begin{align}
            \trace^2(A)\geq\trace\big(A^2\big)\geq\frac{1}{\nAktien}\trace^2(A)
            \label{trace-quadrat}
            \end{align}      
            \item            
            \begin{equation}
            \label{Frob_trace} \|A\|_F=\sqrt{\trace\big(A^2)}\leq           \trace(A)
            \end{equation}
            where $\|A\|_F$ denotes  the Frobenius norm of $A$.

        \end{enumerate}
     \end{lemma}

    \begin{proof}
        The first three properties are standard and we refer to  Horn and Johnson \cite[Chapter 7]{Horn_Johnson}.  The proof of \eqref{trace-symmetric} is given in Wang et al.~\cite[Lemma 1]{wang_kuo_hsu_1986}.       
        \begin{enumerate}
            \setcounter{enumi}{4}

            \item   
            From \eqref{diagona} we have $A=VDV^{\top}$ with an  orthogonal matrix $V$ and
            	$D=\diag(\varrho_1,\cdots,\varrho_d)$. If $A$ is positive definite then $\varrho_{\min}(A)>0$ and the inverse $A^{-1}$ exists, see property \ref{Ainv_symm_posdefinit}. It holds $\trace(A)=\sum_{i=1}^\nAktien \varrho_i(A)\ge \varrho_{\max}(A)$ and 
            	\begin{align*}
            	\trace\big(A^{-1}\big)&=\trace(VD^{-1}V^{\top})=\trace(D^{-1}\big)=\sum_{i=1}^\nAktien \frac{1}{\varrho_i(A)}\ge \frac{1}{\varrho_{\min}(A)}.
            	\end{align*} 
            	The above inequalities together with  \eqref{trace-symmetric} imply  \eqref{trace-product}.

            \item  As above we use $A=VDV^{\top}$ with an  orthogonal matrix $V$ and deduce
             $A^2=VD^2V^\top$ and
            \begin{align}
            \trace\big(A^2\big)&=\trace\big(V^\top VD^2\big)= \trace\big(D^2\big)=\sum_{i=1}^\nAktien\varrho_i^2\geq\frac{1}{\nAktien}\Big(\sum_{i=1}^\nAktien\varrho_i\Big)^2
            =\frac{1}{\nAktien}\trace^2\big(A\big),\nonumber
            \end{align}
            where we have applied the Cauchy-Schwarz inequality. The first inequality in \eqref{trace-quadrat} follows from \eqref{trace-product} with $A=B$.          
            \item
            Let $C=A^2$, then  $C^{ii}=\sum_{k=1}^\nAktien
            (A^{ik})^2$ and
            \begin{align*}
            \trace(A^2)=\trace(C) = \sum_{i=1}^\nAktien
            C^{ii}&=\sum_{i,k=1}^\nAktien (A^{ik})^2 =\|A\|_F^2
            \end{align*}
            yielding the first equality. The inequality follows from \eqref{trace-quadrat}.
        \end{enumerate}
    \end{proof}
\subsection{Proof of Lemma \ref{alpha_Properties}}
\label{proof_lemma_alphaQ}
 For the convenience of the reader we recall the statement of Lemma \ref{alpha_Properties}:\\
{\itshape
    For the function $\alphaQC$ given
    in \eqref{alpha_Q_def} there exist constants
    $a_{\alpha},b_{\alpha}>0$ independent of $\lambda$ and there exists
    $\lambda_0>0$ such that  for all symmetric and positive semidefinite $q\in \R^{d\times d}$
    \begin{align*}
    \trace\big(\alphaQC(q)\big)\leq
    a_{\alpha}-\sqrt{\lambda}\, b_{\alpha} \trace(q),\quad
    \text{for}\quad \lambda\geq \lambda_0.
    \end{align*}
 The above estimate holds for every $a_\alpha>\trace(\Sigma_\mu),$
 $a_\alpha>\trace(\Sigma_\mu),$
 \begin{align}
\nonumber
 b_\alpha < \overline{b}_{\alpha} =\overline{b}_{\alpha}(a_\alpha)&:=2\sqrt{\frac{a_{\alpha}-\trace(\Sigma_\mu)}{\trace(\Gamma)}},~~\\
 \nonumber
 \lambda_0 = {\lambda}_0(a_{\alpha},\beta_{\alpha})&:=
 \bigg(\frac{\nAktien(a_{\alpha}-\trace(\Sigma_\mu))}{
 	2\sqrt{\trace(\Gamma)(a_{\alpha}-\trace(\Sigma_\mu))}-b_{\alpha}\trace(\Gamma)}\bigg)^2.
 \end{align}
 }
   \begin{proof}

    Using the definition of $\alphaQC$ in \eqref{alpha_Q_def}, the linearity of $\trace(\cdot)$ and that $q$
    and $\Sigma_R$ and therefore  $\Sigma_R^{-1}$ are symmetric positive
    definite, and that $\revspeed$ is positive definite we find
    \begin{align}
    \trace\big(\alphaQC(q)\big)&=\trace\Big(\Sigma_{\drift}-\revspeed
    q-q\revspeed^{\top}
    -q\Sigma_R^{-1}q-\lambda q(\Gamma+q)^{-1}q\Big)
    \leq \trace\big(\alphaQCbar(q)\big),
    \label{bound_trace_alpha}\\
    \nonumber
    \text{where}\quad \alphaQCbar(q)&:=\Sigma_{\drift}-\lambda
    q(\Gamma+q)^{-1}q.
    \end{align}
    The inequality follows from properties of
    symmetric positive definite matrices, see \eqref{trace-symmetric},  \eqref{trace-product} and
    \eqref{trace-quadrat}  from which we
    deduce
    \begin{align}
    \trace(\revspeed
    q+q\revspeed^{\top})&=\trace((\revspeed+\revspeed^{\top})q)\geq\varrho_{\min}(\revspeed+\revspeed^{\top})\trace(q)\geq
    0, \nonumber\\
    \trace(q\Sigma_R^{-1}q)&=\trace(q^2\Sigma_R^{-1})\geq\frac{\trace(q^2)}{\trace(\Sigma_R)}\geq
    \frac{\frac{1}{\nAktien}\trace^2(q)}{\trace(\Sigma_R)}
    \geq
    0.\nonumber
    \end{align}
    Here, $\varrho_{\min}(\cdot)$ denotes the
    the smallest eigenvalue of a positive definite
    symmetric matrix, which are all positive. Note that since $\revspeed$ is positive definite  $\revspeed+\revspeed^{\top}$ is symmetric and positive definite.  Further, $q^2$ is symmetric and positive semidefinite and  according to property \ref{Ainv_symm_posdefinit} of Lemma \ref{properties_Gram} $\Sigma_R^{-1}$ is symmetric and positive definite. \\
    Inequality \eqref{bound_trace_alpha} implies that it suffices to
    prove the claim for $\alphaQCbar$, i.e.,
    \begin{align}
    \label{restriction}
    \trace(\alphaQCbar(q))\leq
    a_{\alpha}-\sqrt{\lambda}b_{\alpha}\trace(q)\quad\text{for}\quad
    \lambda\geq \lambda_0.
    \end{align}
    For the proof of \eqref{restriction} we set
    $\varepsilon=\frac{1}{\sqrt{\lambda}}$, $q=\varepsilon z$,
    $a_{\drift}=\trace\big({\Sigma_{\drift}}\big)$ and
    consider the function $H^{\varepsilon}:
    \mathbb{R}^{\nAktien\times\nAktien}\rightarrow \mathbb{R}$ with
    \begin{align}
    H^{\varepsilon}(z)&:=-\trace(\overline{\alpha}^{1/\varepsilon^2}(\varepsilon
    z))+a_{\alpha}-\frac{1}{\varepsilon}b_{\alpha}\trace(\varepsilon
    z)\nonumber\\
    &=\trace(z(\Gamma+\varepsilon
    z)^{-1}z)-b_{\alpha}\trace(z)+a_{\alpha}-a_{\drift}
    \label{function_H}
    \end{align}
  for   $a_{\alpha},    b_{\alpha}, \varepsilon_0>0$ and symmetric and positive semidefinite matrices $z$.
    Below we show that there exist positive constants $a_{\alpha},
    b_{\alpha}, \varepsilon_0$ such that for all $z$ it holds

    \begin{align}
    \label{Positve_H}
    H^{\varepsilon}(z)\geq 0\quad\text{for}\quad \varepsilon\leq
    \varepsilon_0.
    \end{align}
    That inequality implies for
    $z=\frac{1}{\varepsilon}q=\sqrt{\lambda}q$
    \begin{align}
    \nonumber
    0\leq H^{\varepsilon}(z)=H^{\varepsilon}(\sqrt{\lambda}q)=-\trace(\alphaQCbar(q))-\sqrt{\lambda}b_{\alpha}\trace(q)
    +a_{\alpha},
    \end{align}
    and \eqref{bound_trace_alpha}  yields for
    $\lambda\geq\lambda_0=(\frac{1}{\varepsilon_0})^2$
    \begin{align}
    \trace(\alphaQC(q))\leq\trace(\alphaQCbar(q))\leq a_{\alpha}-\sqrt{\lambda}b_{\alpha}\trace(q),
    \nonumber
    \end{align}
    which proves the assertion.\\
    In the remainder of the proof we show inequality \eqref{Positve_H}.   The matrices  $z$ and $\Gamma$
    are symmetric,  $z$ is  positive semidefinite and $\Gamma$ is  strictly positive definite. Then $\Gamma +\eps z$ is   strictly positive definite and according to properties 1.~and
    3.~of Lemma \ref{properties_Gram}, the matrix
    $(\Gamma+\varepsilon z)^{-1}$ is symmetric and  strictly  positive definite. Further, $z^2$ is symmetric and positive semidefinite. Inequality  \eqref{trace-product} implies
    $\trace(AB)\geq\trace(B)/\trace\big(A^{-1}\big)$
    and with $A=(\Gamma+\varepsilon z)^{-1}$ and $B=z^2$ we find 
    \begin{align*}
    \trace(z(\Gamma+\varepsilon
    z)^{-1}z)=\trace(z^2(\Gamma+\varepsilon
    z)^{-1})\geq\frac{\trace(z^2)}{\trace(\Gamma+\varepsilon
        z)}=\frac{\trace(z^2)}{\trace(\Gamma)+\varepsilon
        \trace(z)}.
    \end{align*}
     Inequality \eqref{trace-quadrat}
    yields $\trace(z^2)\geq
    \frac{1}{\nAktien}\trace^2(z)$, and hence we obtain
    \begin{align}
    \label{H1}
    \trace(z(\Gamma+\varepsilon
    z)^{-1}z)\geq\frac{\frac{1}{\nAktien}\trace^2(z)}{\trace(\Gamma)+\varepsilon
        \trace(z)}=\frac{\trace^2(z)}{\gammah+\varepsilon\nAktien
        \trace(z)},
    \end{align}
    where $\gammah=\nAktien\trace(\Gamma)$. Set
    $x=\trace(z)\geq 0$  and $g(x)=\frac{x^2}{\gammah+\varepsilon\nAktien
        x}-b_{\alpha}x+a_{\alpha}-a_{\drift}$ then \eqref{H1} implies
    \begin{align}
    H^{\varepsilon}(z)\geq    g(x). \nonumber
    \end{align}
    Now it remains to choose constants $a_{\alpha}$, $b_{\alpha}$,
    $\varepsilon_0>0$ such that
    \begin{align}
    \label{H_und_f}
    g(x)\geq 0\quad\text{for all}\quad x\geq
    0\quad\text{and}\quad \varepsilon\leq \varepsilon_0.
    \end{align}
    Let $a_{\alpha}>a_{\drift}=\trace(\Sigma_{\drift})$. Then
    $\overline{a}:=g(0)=a_{\alpha}-a_{\drift}>0$. Since
    $\gammah+\varepsilon\nAktien x>0$ the inequality $g(x)\geq 0$ is
    equivalent to
    \begin{align}
    0\leq f(x):=(\gammah+\varepsilon\nAktien
    x)g(x)=A^{\varepsilon}x^2 +B^{\varepsilon}x+C, \nonumber
    \end{align}
    where $A^{\varepsilon}=1-\varepsilon\nAktien b_{\alpha}$,
    $B^{\varepsilon}=\varepsilon\nAktien\overline{a}-b_{\alpha}\gammah$,
    $C=\gammah\overline{a}$. Let $\varepsilon>0$ be chosen so that
    $A^{\varepsilon}>0$, then
    \[
    f(x)=
    A^{\varepsilon}\Big(x^2 + \frac{B^{\varepsilon}}{A^{\varepsilon}}x+ \frac{C}{A^{\varepsilon}}\Big)=A^{\varepsilon}    \Big(x-\frac{K^{\varepsilon}}{2}\Big)^2+D^{\varepsilon}
    \]

    with
    \[
    K^{\varepsilon}=\frac{B^{\varepsilon}}{A^{\varepsilon}}\quad
    \text{and}\quad
    D^{\varepsilon}:=C-\frac{1}{4}\frac{(B^{\varepsilon})^2}{A^{\varepsilon}}=\frac{4CA^{\varepsilon}-(B^{\varepsilon})^2}{4A^{\varepsilon}}.
    \]
    We choose  $a_{\alpha}>a_{\drift}$, i.e.,
    $\overline{a}=a_{\alpha}-a_{\drift}>0$, then we have $f(x)\geq 0$ if
    $D^{\varepsilon}\geq 0$ or equivalently if
    $$P(\varepsilon):=4CA^{\varepsilon}-(B^{\varepsilon})^2 \geq 0.$$
      $P(\varepsilon)$ is a quadratic function and it holds $P(\varepsilon)= 4\gammah \overline a -(\varepsilon d\overline a +b_\alpha \gammah)^2$, hence $P(0)=4\gammah \overline a-b_\alpha^2 \gammah^2$ and $P$ is decreasing for $\varepsilon>0$. Thus we have to require $P(0)>0$ which gives
    \[
    0< b_{\alpha}\leq
    \overline{b}_{\alpha}=\overline{b}_{\alpha}(a_{\alpha})=2\sqrt{\frac{a_{\alpha}-a_{\drift}}{\gammah}}.
    \]
    Then $P(\eps)\ge 0$ for $\eps\in(0,\eps_0] $ where $\eps_0$ is the positive zero of $P$ given by
    \[
    \varepsilon_0=\varepsilon_0(a_{\alpha},\beta_{\alpha})=\frac{1}{\nAktien(a_{\alpha}-a_{\drift})}\Big(2\sqrt{\gammah(a_{\alpha}-a_{\drift})}-b_{\alpha}\gammah\Big).
    \]
    It is not difficult to check that for
    $\varepsilon<\varepsilon_0$
    it holds
    $$A^{\varepsilon}=1-\varepsilon\nAktien b_{\alpha}>\Big(1-b_{\alpha}\sqrt{\frac{\gammah}{a_{\alpha}-a_{\drift}}}\Big)^2\ge 0.$$    
    Note that for $b_\alpha=\overline{b}_\alpha$ it holds $\eps_0=0$ which is not feasible.
    Hence for $a_{\alpha}>a_{\drift}$,
    $b_{\alpha}\in(0,\overline{b}_{\alpha}(a_{\alpha}))$ and for
    $\varepsilon\leq \varepsilon_0=\varepsilon_0(a_{\alpha},b_{\alpha})$  or equivalently $\lambda \ge \lambda_0=1/\varepsilon_0^2 $
    it holds \eqref{H_und_f} and therefore $H^{\varepsilon}(z)\geq 0$  under the conditions given in  \eqref{constants_ablambda}. This completes the proof.
       \end{proof}

\subsection{Proof of Lemma \ref{alpha_Properties_cont}}
\label{proof_lemma_alphaQ_cont}
For the convenience of the reader we recall the statement of Lemma \ref{alpha_Properties_cont}:\\
{\itshape
	For the function $\alphaQD$ given
	in \eqref{alpha_Q_def_cont} there exist constants
	$a_{\alpha},b_{\alpha}>0$ independent of $\lambda$  such that  for all symmetric and positive
	semidefinite $q\in \R^{d\times d}$
	\begin{align}
	\trace\big(\alphaQD(q)\big)\leq 
	a_{\alpha}-\sqrt{\lambda}\, b_{\alpha} \trace(q),\quad
	\text{for}\quad \lambda>0. 
	\nonumber 
	\end{align}
	The above estimate holds for 
$
	a_\alpha=\trace(\Sigma_\mu)+(d \trace(\overline\Sigma_\myzeta)\mynu)^{-1} \quad \text{and} \quad 
	b_\alpha= 2(d \trace(\overline\Sigma_\myzeta)\sqrt{\mynu})^{-1} 
$
	and every $\mynu>0$.
}
\begin{proof}
	
	Using the definition of $\alphaQD$ in \eqref{alpha_Q_def_cont} and  the linearity of $\trace(\cdot)$ 
	we find
	\begin{align}\label{trace_alphaQD}
	\trace\big(\alphaQD(q)\big)&=\trace\big(\Sigma_{\drift}\big) -\trace\big(\revspeed
	q+q\revspeed^{\top}\big)
	-\trace\big(q{\neu (\Sigma_R^{-1} + \lambda \overline \Sigma_\myzeta^{-1})}q\big).
	\end{align}
	For the second term on the r.h.s. \eqref{trace-symmetric} implies
$	\trace(\revspeed
	q+q\revspeed^{\top})=\trace((\revspeed+\revspeed^{\top})q)\geq \beta \trace(q) 
$
	where $\beta:=\varrho_{\min}(\revspeed+\revspeed^{\top})>0$ is the smallest eigenvalue of $\revspeed+\revspeed^{\top}$. That matrix  is  symmetric and positive definite since   $\revspeed$ is positive definite.
Using \eqref{trace-product} and 	\eqref{trace-quadrat} 
	 we deduce  
	\begin{align*} 
{\neu \trace(q(\Sigma_R^{-1}+\lambda \overline\Sigma_\myzeta^{-1})q)}
	 &\ge \lambda \trace(q \overline\Sigma_\myzeta^{-1}q) =  \lambda \trace(q^2 \overline\Sigma_\myzeta^{-1})\\
	&\geq \lambda\frac{\trace(q^2)}{\trace(\overline\Sigma_\myzeta))}
	\geq 	\lambda\frac{\frac{1}{\nAktien}\trace^2(q)}{\trace(\overline\Sigma_\myzeta))} =\lambda \gammah \trace^2(q)
	\end{align*}
	where $\gammah:= (\nAktien\trace(\overline\Sigma_\myzeta))^{-1}>0$.
	Substituting the above  estimates into \eqref{trace_alphaQD} we obtain 
	\begin{align}
	\label{trace_alphaQD2}
	\trace\big(\alphaQD(q)\big) \le f(\trace(q))\quad \text{with}\quad
	f(x):= a_{\drift} -\beta x -\lambda \gammah x^2, ~~x\ge 0,
	\end{align}
	where we set $a_{\drift}=\trace({\Sigma_{\drift}})$. 
	The quadratic function $f$ is strictly concave, thus for any $x_0\ge 0$ it holds $f(x)\le f(x_0)+f^\prime(x_0)(x-x_0)$. Choosing $x_0=1/\sqrt{\lambda \mynu}$ for some $\mynu>0$ it follows
	$$f(x) \le a_{\drift} + \frac{\gammah}{\mynu} - \sqrt{\lambda} \;\frac{2\gammah}{\sqrt{\mynu}}x 
	= a_\alpha -\sqrt{\lambda} b_\alpha x$$
	where we used the definition of $a_\alpha, b_\alpha$ in \eqref{constants_ablambda_cont}.
		Substituting this estimate into \eqref{trace_alphaQD2} proves the claim.
\end{proof}

\end{appendix}

\bibliographystyle{plain}

\end{document}